\documentclass[twoside,draft]{article}

\usepackage{amsfonts}
\usepackage{stmaryrd}
\usepackage{amssymb}
\usepackage{euscript}
\usepackage{amsthm}
\usepackage{amsmath}
\usepackage{amscd}
\usepackage{latexsym}
\usepackage{mathrsfs}
\usepackage{graphicx}
\usepackage{color}
\usepackage{dsfont}
\usepackage{bm}

\newtheorem{theorem}{Theorem}[section]
\newtheorem{lemma}{Lemma}[section]
\newtheorem{proposition}{Proposition}[section]

\newtheorem{remark}{Remark}[section]

\newtheorem{definition}{Definition}[section]
\newtheorem{assumption}{Assumption}[section]
\def\proof{\noindent {\it Proof. $\, $}}
\def\endproof{\hfill $\Box$} 

\newcommand{\Xhh}{X^h}

\newcommand{\psup}{\overline{p}}
\newcommand{\pinf}{\underline{p}}
\newcommand{\vsup}{\overline{v}}
\newcommand{\vinf}{\underline{v}}

\newcommand{\Kpp}{K}
\newcommand{\kpp}{k}

\newcommand{\epsi }{\varepsilon }
\newcommand{\wt }{\widetilde }
\newcommand{\wh }{\widehat }

\newcommand{\whphi }{\widehat{\phi }}

\newcommand{\taus }{\tau^{*}}
\newcommand{\tausi }{\tau^{*,i}}

\newcommand{\taush }{\tau^{*,h}}
\newcommand{\Ma }{\cM^p (\cC^a )}

\newcommand\I{\mathds{1}}

\newcommand\Vben{V^b}

\def\phi{\varphi }

\def\P{{\mathbb P}}

\newcommand{\hHh}{H}
\newcommand{\cEgH}{\cE^{g,H}}
\newcommand{\cEgi}{\cE^{g,i}}

\newcommand{\cEgh}{\cE^{g,h}}

\def\cB{{\mathcal B}}
\def\cC{{\mathcal C}}
\def\cT{{\mathcal T}}
\def\cG{{\mathcal G}}
\def\cH{{\mathcal H}}
\def\cE{{\mathcal E}}

\def\cM{{\mathcal M}}
\def\cS{{\mathcal S}}

\def\cP{{\mathcal P}}

\def\rr{{\mathbb R}}

\def\gg{{\mathbb G}}

\newcommand{\Keywords}[1]{\par\noindent{\small{\bf Keywords\/}: #1}}
\newcommand{\Class}[1]{\par\noindent{\small{\bf Mathematics Subjects Classification (2010)\/}: #1}}

\title{{\Large \bf ARBITRAGE-FREE PRICING OF AMERICAN OPTIONS IN NONLINEAR MARKETS} \vskip 45 pt }

\author{Edward Kim$\,^{b}$, Tianyang Nie$\,^{a}$\footnote{The research of T. Nie and M. Rutkowski was supported by the DVC Research Bridging Support Grant {\it Non-linear Arbitrage Pricing of Multi-Agent Financial Games}. The work of T. Nie was supported by the National Natural Science Foundation of China (No. 11601285) and the Natural Science Foundation of Shandong Province (No. ZR2016AQ13).} \ and Marek Rutkowski$\,^{b,c}$ \\ \\
\\$^{a\,}$School of Mathematics, Shandong University,\\ Jinan, Shandong
250100, China\\ \\  $^{b\,}$School of Mathematics and Statistics, University of Sydney
\\ Sydney, NSW 2006, Australia\\ \\ $^{c\,}$Faculty of Mathematics and Information Science,
Warsaw University of Technology, \\ 00-661 Warszawa, Poland \\ }

\date{\vskip 35 pt \today \vskip 30 pt}

\begin{document}

\maketitle

\begin{abstract}
We re-examine and extend the findings from the recent paper by Dumitrescu et al. \cite{DQS2018} who studied American
options in a particular market model using the nonlinear arbitrage-free pricing approach developed in El Karoui and Quenez \cite{EQ1997}. In the first part, we provide a detailed study of unilateral valuation problems for the two counterparties in an American-style contract within the framework of a general nonlinear market. We extend results from Bielecki et al.~\cite{BCR2018,BR2015} who examined the case of a European-style contract. In the second part, we present a BSDE approach, which is used to establish more explicit pricing, hedging and exercising results when solutions to reflected BSDEs have additional desirable  properties.
\vskip 20 pt
\vskip 10 pt
\Class{91G40,$\,$60J28}
\end{abstract}

\newpage



\section{Introduction} \label{sec0}

Unlike contracts of a European style, contracts of American style are asymmetric between the two counterparties
(hereafter referred to as the {\it issuer} and the {\it holder}) not only due to the opposite directions of contractual cash flows,
but also due to the fact that only one party, the holder of an American option, has the right to {\it exercise} (that is, to stop and settle) an American contract before its expiration date, which we invariably denote as $T$. The issues of arbitrage-free pricing and rational exercising of American options within the framework of a linear market model (usually for the classical Black and Scholes model, but possibly with trading constraints, such as: no borrowing of cash or no short-selling of shares) have been studied in numerous papers, to mention just a few: Bensoussan~\cite{B1984}, El Karoui et al. \cite{EPAQ1997}, Jaillet et al. \cite{JLL1990}, Karatzas~\cite{K1988}, Karatzas and Kou~\cite{KK1998}, Kallsen and K\"uhn~\cite{KK2004},
Klimsiak and Rozkosz \cite{KR2016}, Myneni~\cite{M1992} and Rogers \cite{R2002}.

The goal of this work is to re-examine and extend the findings from the recent paper by Dumitrescu et al.~\cite{DQS2018} who studied American options within the framework of a particular imperfect market model with default using the nonlinear arbitrage-free pricing approach developed in El Karoui and Quenez~\cite{EQ1997}.  In contrast to \cite{DQS2018} (see also \cite{DQS2017} for the case of game options), we place ourselves within the setup of a general nonlinear arbitrage-free market, as introduced in Bielecki et al.~\cite{BCR2018,BR2015}, and we examine general properties of fair and acceptable unilateral prices. We also obtain more explicit results regarding pricing, hedging and break-even times for the issuer and rational exercise times for the holder using a BSDE approach without being specific about the dynamics of underlying assets, but by focusing instead on pertinent general features of solutions to reflected BSDEs.

Let us introduce some notation for a generic nonlinear market model. Let $(\Omega, \cG, \gg,\P )$ be a filtered probability space satisfying the usual conditions of right-continuity and completeness, where the filtration $\gg=(\cG_t)_{t \in [0,T]}$ models the flow of information available to all traders. For convenience, we assume that the initial $\sigma$-field $\cG_0$ is trivial. Moreover, all processes introduced in what follows are implicitly assumed to be $\gg$-adapted and, as usual, any semimartingale is assumed to be c\`adl\`ag. For simplicity of notation, we assume throughout that the trading conditions are identical for the issuer and the holder, although this assumption can be relaxed without any difficulty since we examine unilateral valuation and hedging problems.
Let $\cT=\cT_{[0,T]}$ stand for the class of all $\gg$-stopping times taking values in $[0,T]$. We adopt the following definition
of an American contingent claim.

By convention, all cash flows of a contract are described from the perspective of the issuer. Hence when a cash
flow is positive for the issuer, then the cash amount is paid by the holder and received by the issuer. Obviously, if a cash flow
is negative for the issuer, then the cash amount is transferred from the issuer to the holder.
For instance, when dealing with the classical case of an American put option written on a stock $S$, we assume that the payoff to the issuer  (respectively, the holder) equals $\Xhh_{\tau}=-(K-S_{\tau})^+$  (respectively, $\Xhh_{\tau}=(K-S_{\tau})^+$) if the option is exercised at time $\tau $ by the holder. This is formalized through the following definition of an {\it  American contingent claim.}
Note that the superscript $h$ in $\Xhh$ is used to emphasize that only the holder has the right to exercise and hence to stop the contract;
this can be contrasted with the case of a game option (see, e.g., Kifer~\cite{KY2000,KY2013}  and the references therein) where the covenants stipulate that both parties may exercise and hence stop the contract.

\begin{definition} \label{def0.1}
{\rm An {\it  American contingent claim} with the $\gg$-adapted, c\`adl\`ag payoff process $\Xhh$ is a contract between the issuer and the holder where the holder  has the right to exercise the contract by selecting a $\gg $-stopping time $\tau\in\cT_{[0,T]}$. Then the issuer `receives' the amount $\Xhh_{\tau}$ or, equivalently, `pays' to the holder the amount of $- \Xhh_\tau$ at time $\tau$ where the $\gg$-adapted payoff process $\Xhh_t,\,t\in[0,T]$ is specified by the contract. Note that we do not make any a priori assumptions about the sign of the payoff process $\Xhh$, so it can be either positive or negative, in general.}
\end{definition}

More generally, an {\it American contract} is formally identified with a triplet $\cC^a =(A,\Xhh,\cT )$ where a $\gg$-adapted,
c\`adl\`ag stochastic process $A$, which is predetermined by the contract's clauses, represents the {\it cumulative cash flows} from time 0 till the contract's maturity date $T$. In the financial interpretation, the process $A$ is assumed to model all the cash flows of a given American contract, which are either paid out from the issuer's wealths or added to his wealth via the value process of his portfolio of traded assets.

By symmetry, an analogous interpretation applies to the holder of an American contract and, obviously, any amount received (respectively, paid) by one of the parties is paid (respectively, received) by the other one. Note, however, that the price of the contract $\cC^a$, which is exchanged at its initiation (by convention, at time 0), is not included in the process $A$ so that we set $A_0=0$. This convention is motivated by the fact that the contract's price before the deal is made is yet unspecified and thus it needs to be determined through negotiations between the counterparties and we will argue that unilateral pricing does not yield a common value for the initial price of an American contract, in general.

When examining the valuation of an American contract at any time  $t\in[0,T]$, we implicitly assume that it has not yet been exercised and thus the set of exercise times available at time $t$ to its current holder is the class $\cT_{[t,T]}$ of all $\gg$-stopping times taking values in $[t,T]$. In principle, one could consider two alternative conventions regarding the payoff upon exercise: either
\hfill \break
(A.1)  the cash flow upon exercise at time $t$ equals $A_t-A_{t-}+\Xhh_t$ or \hfill \break
(A.2)  if a contract is exercised at time $t$, then the cash flow $A_t-A_{t-}$ is waived, so the only cash flow occurring at time $t$ is $\Xhh_t$.
\hfill \break Unless explicitly stated otherwise, we work under covenant (A.1) and we acknowledge that the choice of a particular settlement rule may result in a different price for an American contract $\cC^a$, in general. Of course, this choice is immaterial when the process $A$ is continuous or, simply, when it vanishes, so that the contract reduces to a pair $(\Xhh,\cT )$.

In the classical linear market model, if the terminal payoff of a contract is negative (or positive), then it is easy to recognize
whether the fair price for a given party should be positive (or negative) and thus it is less important to keep track of signs of cash flows. In contrast, in a nonlinear setup it may occur, for instance, that each party would like to sell a contract for a positive
respective price and thus it is not reasonable to make any a priori assumptions about the signs of prices.

The paper is organized as follows. In Section \ref{sec1}, we work in an abstract nonlinear setup, meaning that we only make fairly general assumptions about the nonlinear dynamics of the wealth process of self-financing strategies. The main postulates of that kind are the monotonicity properties of the wealth (see Assumptions \ref{ass1.1} and \ref{ass1.2}). We examine general properties of fair and profitable prices for the two counterparties, the issuer and the holder. In particular, we built upon papers by Bielecki et al.~\cite{BCR2018,BR2015} where the arbitrage-free valuation of European contingent claims in nonlinear markets was examined.
In Section \ref{sec4}, we re-examine and extend a BSDE approach to the valuation of American options in nonlinear market initiated by
El Karoui and Quenez~\cite{EQ1997} and continued in Dumitrescu et al.~\cite{DQS2018}. Our main goal is to show that unilateral acceptable prices for an American contract $\cC^a$ can be characterised in terms of solutions to reflected BSDEs driven by a multi-dimensional continuous semimartingale $S$. For the sake of concreteness, we postulate in Section \ref{sec4} that the wealth process $V=V(y,\xi ,A)$ satisfies
\begin{equation} \label{uyeq3x}
V_t=y-\int_0^t g(u,V_u,\xi_u)\,du+\int_0^t\xi_u\,dS_u+A_t,
\end{equation}
where $y \in \rr$ is the initial wealth at time 0 of a dynamic portfolio $\xi$ (recall that $A_0=0$). However, in order to
keep the notation comprehensive, we do not explicitly specify the process $S$ or generator $g$ but instead we make
assumptions about solutions to the SDE \eqref{ueq3x} and the BSDE  \eqref{nhBSDE}.

\section{Unilateral Fair and Acceptable Prices} \label{sec1}

Let $\cM=(\cB,\cS,\Psi)$ be a market model which is arbitrage-free with respect to European contracts in the sense of  Bielecki et al.~\cite{BCR2018,BR2015}. Here $\Psi$ stands for the class of all {\it admissible} trading strategies and $\Psi(y,D)$ denotes the class of all admissible trading strategies from $\Psi$ with initial wealth $y\in\rr$ and with external cash flows $D$. For any trading strategy $\phi\in\Psi (y,D)$, we denote by $V(y,\phi,D)$ the {\it wealth process} of $\phi$. Obviously, the equality $V_0(y,\phi,D)=y$ holds for all $y\in\rr$ and any strategy $\phi$. It is assumed throughout that the processes $D,\Xhh$ and the wealth process $V(y,\phi,D)$ are c\`adl\`ag and $\gg$-adapted. We will gradually make more assumptions about the dynamics of wealth processes.

\subsection{Benchmark Wealth} \label{sec1.1}

An important feature of the nonlinear arbitrage-free approach is the concept of the {\it benchmark wealth} $\Vben (x)$ with respect to which arbitrage opportunities of a given trader are quantified and assessed. As in \cite{BCR2018,BR2015}, a simple and natural candidate for the benchmark wealth can be given by the equality $\Vben (x)=V^0(x)$ where, for an arbitrary initial endowment $x\in\rr$ of a trader, we set for all $t\in[0,T]$
\[
V_t^0(x):=xB_t^{0,l}\I_{\{x\geq 0\}}+ xB_t^{0,b}\I_{\{x<0\}}
\]
where the risk-free \textit{lending} (respectively, {\it borrowing}) \textit{cash account} $B^{0,l}$ (respectively, $B^{0,b}$) is used for unsecured lending (respectively, borrowing) of cash. Note that $V^0(x)$ represents the wealth process of a trader who decided at time 0 to keep his initial cash endowment $x$ in either the lending (when $x \geq 0$) or the borrowing (when $x<0$) cash account and who is not involved in any other trading activities between the times 0 and $T$. By convention, the quantities $x_1$ and $x_2$ represent initial endowments of the issuer and the holder and thus the processes $V^0(x_1)$ and $V^0(x_2)$ are their respective benchmark wealths.
Since the idea of the benchmark wealth is immaterial for valuation in linear market models, it is rarely encountered in works on arbitrage-free valuation, although it corresponds to the well-known economic concept of {\it opportunity costs}.
It thus important to make few comments on its relevance, scope and limitations.

On the one hand, we acknowledge that the idea of unilateral pricing based on an initial portfolio of a trader can be rejected as unrealistic. It is thus worth stressing that even when we set $x_1=x_2=0$ so that $V_t^0(x_1)=V_t^0(x_2)=0$ for all $t\in[0,T]$, meaning that the initial endowments of traders are completely ignored, the asymmetry in their respective unilateral prices will show up
as a result of nonlinearity of the wealth dynamics, so that this assumption would not give an essential simplification of our results and findings.

On the other hand, it would also be possible to assume that each trader is endowed with an initial portfolio of assets (including also the savings account $B^{0,l}$ and $B^{0,b}$, which means that he could be either a borrower or a lender of cash) with the current market value $y$ at time 0 (ignoring bid-ask spreads and transaction costs). Then $V^b(y)$ could represent the wealth process of his static portfolio. Assume, for instance, that the issuer of an American contract has an initial portfolio of risky assets, denoted as $\alpha_0$, and cash, denoted as $\beta_0$, with the (static) terminal values $V_T(\alpha_0)$ and $V^0_T(\beta_0)$, respectively.  Suppose that the risky portfolio $\alpha_0$ is not used for issuer's hedging purposes but a (positive or negative) amount $x_1$ from the cash account $\beta_0$ is used to establish the hedge. Then the issuer's benchmark wealth $V^b(\alpha_0,\beta_0)$ can be defined as $V^b_t(\alpha_0,\beta_0):=V_t(\alpha_0)+V^0_t(\beta_0)$ for all $t\in[0,T]$ so that, in particular, the initial wealth of the issuer is equal to $y:= V^b_0(\alpha_0,\beta_0)=V_0(\alpha_0)+V^0_0(\beta_0)$. Then the issuer's total wealth, inclusive of the price $p$ for the contract $\cC^a$ entered at time 0 and the hedging strategy $\phi $, equals
$$
V_t(\alpha_0,\beta_0,x_1,p,\phi,A)=V_t(\alpha_0)+V^0_t(\beta_0-x_1)+V_t(x_1+p,\phi,A).
$$
where $x_1$ no longer represents the initial wealth of the issuer, but rather the part of the initial cash endowment used for hedging of $\cC^a$. It is clear that it would be more difficult to analyze the case of a dynamic benchmark portfolio, although in principle this is possible. For the sake of conciseness, we are not going to study these cases in what follows, but we stress that all results in
this work are valid for any specifications of the benchmark wealth processes for the issuer and the holder,
respectively, since it suffices to assume that they are some $\gg$-adapted and c\`adl\`ag stochastic processes.


\subsection{Unilateral Fair Prices} \label{sec1.2}

We consider an extended market model $\Ma $ in which an American contract $\cC^a$ is traded at time~$0$ at some initial price $p$ where $p$ is an arbitrary real number. We first give a preliminary analysis of unilateral fair valuation of an American contract by its issuer who is endowed  with the pre-trading initial wealth $x_1\in\rr$ and the corresponding benchmark wealth process $\Vben (x_1)$.

\subsubsection{Issuer's Fair Prices} \label{sec1.2.1}

We start by introducing the following conditions. Since the process $A$ is fixed throughout, to alleviate notation, we will frequently write $V(x_1+p,\phi)$ instead of $V(x_1+p,\phi,A)$ when dealing with the issuer. By the same token, we will later write $V(x_2-p,\psi)$ instead of $V(x_2-p,\psi,-A)$ when examining trading strategies of the holder.

\begin{definition} \label{def1.1}
{\rm We say that a triplet $(p,\phi,\tau)\in\rr\times\Psi(x_1+p,A)\times\cT$ satisfies:
\[
\begin{array}[c]{lll}
&\text{(AO)}& \Longleftrightarrow \ V_{\tau}(x_1+p,\phi)+\Xhh_{\tau}\geq \Vben_{\tau}(x_1)
 \ \text{\rm and } \P \big(V_{\tau}(x_1+p,\phi)+\Xhh_{\tau}>\Vben_{\tau}(x_1)\big)>0,\medskip \\
&\text{(SH)}& \Longleftrightarrow \  V_{\tau}(x_1+ p,\phi)+\Xhh_{\tau} \geq  \Vben_{\tau}(x_1),  \medskip\\
&\text{(BG$_{\epsi}$)}&\Longleftrightarrow \ V_{\tau}(x_1+p,\phi)+\Xhh_{\tau} \leq \Vben_{\tau}(x_1)+\epsi,  \medskip\\
&\text{(BE)}&\Longleftrightarrow \ V_{\tau}(x_1+p,\phi)+\Xhh_{\tau}=\Vben_{\tau}(x_1),  \medskip \\
&\text{(NA)}&\Longleftrightarrow \ \text{\rm either } V_{\tau}(x_1+p,\phi)+\Xhh_{\tau}=\Vben_{\tau}(x_1)
  \ \text{\rm or } \P \big(V_{\tau}(x_1+p,\phi)+\Xhh_{\tau}<\Vben_{\tau}(x_1)\big)>0.
\end{array}
\]}
\end{definition}

Let us first explain the meaning of acronyms appearing in Definition \ref{def1.1}: (AO) stands for {\it arbitrage opportunity}, (SH) for  {\it superhedging}, (BG) for {\it bounded gain}, (BE) for {\it break-even} and (NA) for {\it no-arbitrage}. For a more detailed
explanation of each of these properties, see Definitions \ref{def1.2}--\ref{def1.6}.

For brevity, we write $(p ,\phi ,\tau) \in$ (AO) if a triplet $(p ,\phi ,\tau)$ satisfies condition (AO); an analogous convention will be applied to other conditions introduced in Definition \ref{def1.1}.

If property (SH) is satisfied by a triplet $(p,\phi,\tau)$, then we say that an {\it issuer's superhedging at time} $\tau $ arises. Note that, from the optional section theorem, condition (SH) holds for a pair $(p,\phi)\in \rr\times\Psi(x_1+p,A) $ and {\it all} $\tau\in\cT $ if and only if $(p,\phi)$ is such that the inequality $V_t (x_1+p ,\phi)+\Xhh_t \geq \Vben_t(x_1)$ is valid for every $t \in [0,T]$. This simple observation justifies the following definition of property (SH) for a pair $(p,\phi)$.

\begin{definition} \label{def1.2}
{\rm We say that a pair $(p,\phi)\in \rr\times\Psi(x_1+p,A)$ satisfies (SH) (briefly, $(p,\phi )\in \, $(SH)) if the inequality $V_t (x_1+p,\phi)+\Xhh_t \geq \Vben_t(x_1)$ holds for every $t \in [0,T]$. Then $(p,\phi)$ is called an {\it issuer's superhedging strategy} in the extended market $\Ma $.}
\end{definition}

Property (AO) of a triplet $(p,\phi,\tau)$ is referred to as the {\it issuer's strict superhedging} (or the {\it issuer's arbitrage opportunity}) at time $\tau $.

\begin{definition} \label{def1.3}
{\rm We say that a pair $(p,\phi)\in \rr\times\Psi(x_1+p,A)$ satisfies condition (AO) if a triplet $(p,\phi ,\tau)$ complies with condition (AO) for every $\tau\in\cT$. Then we also say that $(p,\phi)$ creates an {\it  issuer's arbitrage opportunity} in the extended market $\Ma $.}
\end{definition}

The following lemma is an immediate consequence of Definition \ref{def1.3}.

\begin{lemma} \label{lem1.1}
If a pair $(p,\phi)\in \rr\times\Psi(x_1+p,A)$ is such that $V_t (x_1+p,\phi)+\Xhh_t> \Vben_t(x_1)$ for every $t \in [0,T]$,
then $(p,\phi)$ fulfills (AO) and thus an issuer's arbitrage opportunity arises in the extended market $\Ma $.
\end{lemma}

We say that {\it  no issuer's arbitrage arises for} $(p,\phi)$ {\it  at} $\tau $ if $(p,\phi ,\tau)$ satisfies condition (NA).
It readily seen that if a triplet $(p,\phi ,\tau)$ fails to satisfy (NA), then it fulfills (AO) and thus an issuer's arbitrage opportunity arises at time $\tau$ for the issuer's strategy $(p,\phi)$. By convention, we henceforth set $\inf \emptyset=\infty$ and $\sup \emptyset=-\infty $. Note that the superscript $f$ stands here for {\it fair} and $i$ for {\it issuer}.

\begin{definition}\label{def1.4} {\rm
We say that $p^{f,i} (x_1,\cC^a)$ is an {\it issuer's fair price} for $\cC^a$ if no issuer's arbitrage opportunity $(p,\phi)$
may arise in $\Ma $ when $p=p^{f,i} (x_1,\cC^a)$. Hence the set of issuer's fair prices equals
 \[
\cH^{f,i}(x_1):=\big\{ p\in\rr \,|\ \forall \,\phi\in\Psi (x_1+p,A) \, \exists \,\tau\in\cT\!: (p,\phi,\tau) \in \text{(NA)} \big\}
\]
and the upper bound for issuer's fair prices is given by
\begin{align} \label{eq7a}
\psup^{f,i}(x_1,\cC^a):=\sup \big\{p\in\rr \,|\ p \text{\ is an issuer's fair price for\ } \cC^a\big\}=\sup \,\cH^{f,i}(x_1).
\end{align}
If the equality $\psup^{f,i}(x_1,\cC^a)=\max \, \cH^{f,i}(x_1)$ holds (that is, whenever $\psup^{f,i}(x_1,\cC^a) \in \cH^{f,i}(x_1)$), then $\psup^{f,i}(x_1,\cC^a)$ is denoted as $\wh{p}^{f,i}(x_1,\cC^a)$ and called the {\it issuer's maximum fair price} for $\cC^a$.
} \end{definition}

To alleviate notation, the variables $(x_1,\cC^a)$ will be sometimes suppressed and thus we will write $p^{f,i} ,\,
\psup^{f,i},\, \wh{p}^{f,i},$ etc., rather than $p^{f,i} (x_1,\cC^a),\, \psup^{f,i}(x_1,\cC^a),\, \wh{p}^{f,i}(x_1,\cC^a),$ etc.

\begin{assumption}  \label{ass1.1}
{\rm  The {\it forward monotonicity} property holds:
for all $x, p\in\rr,\, \phi\in\Psi (x+p,A)$ and $p'>p$ (respectively, $p'<p$),
there exists a trading strategy $\phi'\in\Psi (x+p',A)$ such that $V_t (x+p',\phi' ) \geq V_t (x+p,\phi)$
(respectively, $V_t (x+p',\phi' ) \leq V_t (x+p,\phi)$) for every $t \in [0,T]$.}
\end{assumption}

\begin{lemma} \label{lem1.2}
Let Assumption \ref{ass1.1} be satisfied.  If $ p\in\cH^{f,i}(x_1)$, then for any $p' <p$ we have that $p'\in\cH^{f,i}(x_1)$.
Hence if $\cH^{f,i}(x_1) \ne \emptyset$, then either $\cH^{f,i}(x_1)=(-\infty,\psup^{f,i}\,]=(-\infty,\wh{p}^{f,i}]$ for some
$\psup^{f,i}\in\rr$ or $\cH^{f,i}(x_1)=(-\infty,\psup^{f,i})$ where $\psup^{f,i}\in\rr \cup \{ \infty\}.$
\end{lemma}

\proof
We argue by contradiction. If $\cH^{f,i}(x_1)=\emptyset$, then $\psup^{f,i}=- \infty $. Let us now consider the case where $\cH^{f,i}(x_1) \ne \emptyset$. Assume that $ p\in\cH^{f,i}(x_1)$ and a number $p'< p$ is not an issuer's fair price. Then there exists $\phi'\in\Psi (x_1+p',A)$ such that $(p',\phi' ,\tau)$ fulfills (AO) for every $\tau\in\cT$. Consequently, by Assumption \ref{ass1.1}, there exists $\phi \in\Psi (x_1+p,A)$ such that a triplet $(p,\phi ,\tau)$ complies with (AO) for every $\tau\in\cT$. This clearly contradicts the assumption that $p$ belongs to $\cH^{f,i}(x_1)$ and thus we conclude that the asserted properties are valid.
\endproof

%
%

The {\it bounded gain} condition (BG$_{\epsi}$) stipulates that issuer's gains associated with a triplet $(p,\phi,\tau )$ are bounded from above by $\epsi $. It leads to the following definition of an issuer's superhedging cost with negligible gain.

\begin{definition} \label{def1.5} {\rm
We say that $p \in \rr$ is an {\it issuer's superhedging cost with negligible gain}  for $\cC^a$ if for every $\phi\in \Psi(x_1+p,A)$ such that condition (SH) is satisfied by $(p,\phi)\in \rr\times\Psi(x_1+p,A)$ and for every $\epsi>0$, there exists a $\tau^{\epsi} \in \cT$ such that $V_{\tau^{\epsi}}(x_1+p,\phi)+ \Xhh_{\tau^{\epsi}} \leq \Vben_{\tau^{\epsi}}(x_1)+ \epsi$. The set of all issuer's superhedging costs
 with negligible gain  for $\cC^a$ is denoted by $\cH^{n,i}(x_1)$, that is,
\[
\cH^{n,i}(x_1):=\big\{p\in\rr\,|\ \forall \,(p,\phi)\in \rr \times \Psi(x_1+p,A) \mbox{\rm \ satisfying (SH)} \forall\,\epsilon>0\, \exists \,\tau^{\epsi}\in\cT\!:(p,\phi,\tau^{\epsi})\in \text{(BG$_{\epsi}$)}\big\}.
\] }
\end{definition}

\begin{remark} \label{rem1.1}
{\rm Note that the issuer's superhedging cost with negligible gain is not necessarily unique. For example, assume that a pair $(p,\phi)\in \rr\times\Psi(x_1+p,A)$  satisfies (SH) so that $V_{\tau}(x_1+ p,\phi)+\Xhh_{\tau} \geq  \Vben_{\tau}(x_1)$ for all $\tau\in \cT$ and, in addition, $V_{T-}(x_1+p,\phi)+ \Xhh_{T-}= \Vben_{T-}(x_1)$. It is possible to suppose that there exists a $\delta>0$ such that for every $\phi' \in \rr\times\Psi(x_1+p+\delta,A) $ such that $(p,\phi')\in\, $ (SH), we have $V_{T-}(x_1+p,\phi)+ \Xhh_{T-} =V_{T-}(x_1+p+\delta,\phi')+ \Xhh_{T-}= \Vben_{T-}(x_1)$ and $V_{T}(x_1+p+\delta,\phi')+ \Xhh_{T}>V_{T}(x_1+p,\phi)+ \Xhh_{T}$. Then it is obvious that both  $p$ and $p+\delta$ are issuer's superhedging costs with negligible gain.}
\end{remark}

Let us finally introduce a stopping time related to the {\it break-even} condition (BE) introduced in Definition \ref{def1.1}.

\begin{definition} \label{def1.6} {\rm
If condition (BE) is satisfied by $(p,\phi,\tau)\in\rr\times\Psi (x_1+p,A)\times\cT$, then a stopping time $\tau\in\cT $ is called an {\it issuer's break-even time} for the pair $(p,\phi)\in\rr\times\Psi (x_1+p,A)$.}
\end{definition}

Note that even when the pair $(p,\phi)$ is fixed, the uniqueness of an issuer's break-even time $\tau$ is not ensured, in general.
Obviously, any issuer's break-even time can be formally classified as one of the exercise times available to the holder of $\cC^a$
but, as we will argue in what follows, an issuer's break-even time is unlikely to also be a {\it rational exercise time} for the holder. This is due to the fact that it may not actually always be advantageous for the holder to exercise at a stopping time that causes the issuer to break even or prohibits the issuer's arbitrage opportunities.  Firstly, usually the holder is not informed about the issuer's trading strategy. Secondly, the holder should be behaving in a rational way for his own payoff and hedging abilities.  A holder's rational exercise time can be typically identified with a particular instance of a {\it holder's break-even time}, which is introduced in Definition \ref{def1.8}. The earliest issuer's break-even time will be later denoted as $\tausi$, whereas for the earliest holder's rational exercise time we will use the symbol $\taush $.

\subsubsection{Holder's Fair Prices} \label{sec1.2.2}

Let us now analyze the holder's fair pricing problem for $\cC^a$.  We assume that he is endowed with the pre-trading initial wealth $x_2\in\rr$ and his computation refers to the benchmark wealth process $\Vben (x_2)$. Recall that we write $V(x_2-p,\psi) := V(x_2-p,\psi,-A)$ when there is no danger of confusion.

\begin{definition} \label{def1.7}
{\rm We say that $(p,\psi ,\tau)\in\rr\times\Psi(x_2-p,-A)\times\cT$ satisfy:
\[
\begin{array}[c]{lll}
&\text{(AO$'$)}&\Longleftrightarrow \  V_{\tau}(x_2-p,\psi)-\Xhh_{\tau} \geq \Vben_{\tau}(x_2)\
\text{\rm and } \P \big(V_{\tau}(x_2-p,\psi)-\Xhh_{\tau}>\Vben_{\tau}(x_2)\big)>0,\medskip \\
&\text{(SH$'$)} &\Longleftrightarrow \ V_{\tau}(x_2-p,\psi)-\Xhh_{\tau} \geq \Vben_{\tau}(x_2), \medskip\\
&\text{(BL$_{\epsi}'$)}&\Longleftrightarrow \ V_{\tau}(x_2-p,\psi)- \Xhh_{\tau} \geq \Vben_{\tau}(x_2) - \epsi , \medskip\\
&\text{(BE$'$)} &\Longleftrightarrow \ V_{\tau}(x_2-p,\psi)-\Xhh_{\tau}=\Vben_{\tau}(x_2), \medskip \\
&\text{(NA$'$)}&\Longleftrightarrow \ \text{\rm either } V_{\tau}(x_2-p,\psi )-\Xhh_{\tau}=\Vben_{\tau}(x_2)\
\text{\rm or } \P \big(V_{\tau}(x_2-p,\psi)-\Xhh_{\tau} < \Vben_{\tau}(x_2)\big)>0.
\end{array}
\] }
\end{definition}

Property (AO$'$) (respectively, (SH$'$)) is called the {\it strict superhedging} (respectively, {\it superhedging}) condition for the holder.
Note that the {\it bounded loss} property (BL$_{\epsi}'$) means that the holder's losses are bounded from below by $-\epsi $.
Condition (BE$'$) leads to the following definition.

\begin{definition}  \label{def1.8}
{\rm If the equality  $V_{\tau' }(x_2-p ,\psi )-\Xhh_{\tau' }=\Vben_{\tau' }(x_2)$ holds, then a stopping time $\tau'\in\cT$ is called a {\it holder's break-even time} for the pair $(p,\psi )\in\rr\times\Psi (x_2-p,-A)$.}
\end{definition}

The concept of a holder's arbitrage opportunity reflects the fact that the holder has the right to exercise an American contract, that is, to conveniently choose a stopping time $\tau $ at which the contract is settled and terminated. Specifically, a {\it  holder's arbitrage opportunity in} $\Ma$ is a triplet $(p,\psi,\tau)\in \rr\times\Psi (x_2-p,-A)\times\cT$ satisfying condition (AO$'$).

For a triplet $(p,\psi,\tau)\in\rr\times\Psi (x_2-p,-A)\times\cT$, we say that {\it  no holder's arbitrage arises for} $(p,\psi)$ {\it  at time} $\tau$ if $(p,\psi,\tau)$ fulfills (NA$'$). It is easily seen that a triplet $(p, \psi,\tau)\in \rr\times\Psi (x_2-p,-A)\times\cT $ fails to satisfy (AO$'$) if and only if it satisfies (NA$'$).

\begin{definition}\label{def1.9} {\rm
We say that $p^{f,h}(x_2,\cC^a)$ is a {\it holder's fair price} for $\cC^a$ if no holder's arbitrage opportunity $(p,\psi,\tau)$ may arise in the extended market $\Ma$ when $p=p^{f,h}(x_2,\cC^a)$. Hence the set of holder's fair prices equals
\[
\cH^{f,h}(x_2):=\big\{ p\in\rr \,|\  \forall \, (\psi ,\tau)\in\Psi (x_2-p,-A)\times\cT\!: (p, \psi ,\tau)\in \text{(NA$'$)} \big\}
\]
and the lower bound for the holder's fair prices is given by
\begin{align} \label{eq8a}
\pinf^{f,h}(x_2,\cC^a):=\inf\big\{p\in\rr\,|\  p \text{\ is a holder's fair price for\ } \cC^a\big\}=\inf\,\cH^{f,h}(x_2).
\end{align}
If the equality $\pinf^{f,h}(x_2,\cC^a)=\min\,\cH^{f,h}(x_2)$ holds, then $\pinf^{f,h}(x_2,\cC^a)$ is denoted as $\breve{p}^{f,h}(x_2,\cC^a)$ and called the {\it holder's minimum fair price} for $\cC^a$.}
\end{definition}

\begin{lemma} \label{lem1.3}
Let Assumption \ref{ass1.1} be satisfied for $-A$. If $ p\in\cH^{f,h}(x_2)$, then for any $p'>p$ we have that $p'\in\cH^{f,h}(x_2)$.
Therefore, if $\cH^{f,h}(x_2) \ne \emptyset$, then either $\cH^{f,h}(x_2)=[\,\pinf^{f,h},\infty )=[\,\breve{p}^{f,h},\infty )$
or $\cH^{f,h}(x_2)=(\pinf^{f,h},\infty)$.
\end{lemma}

%

\begin{definition} \label{def1.10} {\rm
We say that $p \in \rr$ is a {\it holder's cost with negligible loss} if for every $\epsi>0$ there exists a pair $(\psi,\tau)\in\Psi (x_2-p,-A)\times\cT$ such that $V_{\tau}(x_2-p,\psi)- \Xhh_{\tau} \geq \Vben_{\tau}(x_2)-\epsi $. The set of all holder's costs with negligible loss for $\cC^a$ is denoted by $\cH^{n,h}(x_2)$, that is,
\[
\cH^{n,h}(x_2):=\big\{p\in\rr\,|\  \forall \, \epsilon >0 \, \exists \, (\psi,\tau)\in\Psi (x_2-p,-A)\times\cT\!:(p,\psi,\tau)\in \text{(BL$_{\epsi}'$)}\big\}.
\] }
\end{definition}

\subsection{Superhedging Costs}  \label{sec1.3}

The concepts of a (strict) superhedging strategy and the associated cost for the issuer and the holder are fairly standard.
For the issuer, they are based on conditions (SH) and (AO), respectively, whereas for the holder they hinge on conditions (SH$'$) and (AO$'$), respectively. This means that for the issuer we need to impose some conditions that are valid for every $\tau\in\cT$, whereas for the holder it suffices to postulate that the analogous conditions are satisfied for some $\tau\in\cT$.

\subsubsection{Issuer's Superhedging Costs} \label{sec1.3.1}

We first introduce the notion of the lower bound for issuer's strict superhedging costs.

\begin{definition} \label{def1.11}
{\rm The lower bound for issuer's strict superhedging costs for $\cC^a$ is given by $\pinf^{a,i}(x_1,\cC^a)$ $:=\inf \,\cH^{a,i}(x_1)$ where
\[
\cH^{a,i}(x_1):=\big\{ p\in\rr :\,\exists \,\phi\in\Psi(x_1+p,A):\! (p,\phi) \in \textrm{(AO)}\big\}.
\]
If the equality $\pinf^{a,i}(x_1,\cC^a)=\min \,\cH^{a,i}(x_1)$ holds, then it is denoted as $\breve{p}^{a,i}(x_1,\cC^a)$ and called the {\it issuer's minimum strict superhedging cost} for $\cC^a$.}
\end{definition}

It is readily seen that $\cH^{a,i}(x_1)$ is the complement of $\cH^{f,i}(x_1)$ and thus, in view of Lemma \ref{lem1.2},
the equality $\pinf^{a,i}(x_1,\cC^a)=\psup^{f,i}(x_1,\cC^a)$ is satisfied under Assumption \ref{ass1.1}. More precisely, we deal with the following alternative: either
\begin{equation} \label{first}
\cH^{f,i}(x_1)=(-\infty,\wh{p}^{f,i}\,]\ \,\text{and}\ \,\cH^{a,i}(x_1)=(\pinf^{a,i},\infty)
\end{equation}
or
\begin{equation} \label{second}
\cH^{f,i}(x_1)=(-\infty,\psup^{f,i}\,) \ \,\text{and}\ \,\cH^{a,i}(x_1)=[\,\breve{p}^{a,i},\infty).
\end{equation}

\begin{definition} \label{def1.12}
{\rm The lower bound for issuer's superhedging costs for $\cC^a$ is given by $\pinf^{s,i}(x_1,\cC^a)$ $:=\inf \,\cH^{s,i}(x_1)$ where
\[
\cH^{s,i}(x_1):=\big\{p\in\rr :\,\exists\,\phi\in\Psi (x_1+p,A)\!: (p,\phi) \in \textrm{(SH)}\big\}.
\]
If the equality $\pinf^{s,i}(x_1,\cC^a)=\min\,\cH^{s,i}(x_1)$ holds, then $\pinf^{s,i}(x_1,\cC^a)$ is denoted as $\breve{p}^{s,i}(x_1,\cC^a)$ and called the {\it issuer's minimum superhedging cost} for $\cC^a$.}
\end{definition}

It is clear that $\cH^{a,i}(x_1) \subseteq \cH^{s,i}(x_1)$ and thus $\pinf^{s,i}(x_1,\cC^a) \leq \pinf^{a,i}(x_1,\cC^a)$.
In general, it may occur that $\pinf^{s,i}(x_1,\cC^a)<\pinf^{a,i}(x_1,\cC^a)=\psup^{f,i}(x_1,\cC^a)$.
To avoid this problematic situation, we introduce Assumption \ref{ass1.2}, which ensures that $\pinf^{s,i}(x_1,\cC^a)=\pinf^{a,i}(x_1,\cC^a)$
and thus also $\pinf^{s,i}(x_1,\cC^a)=\psup^{f,i}(x_1,\cC^a)$. Obviously, Assumption \ref{ass1.2} is stronger than Assumption \ref{ass1.1}.

\begin{assumption}  \label{ass1.2}
{\rm  The {\it forward strict monotonicity} property holds: for all $x,p\in\rr,\, \phi\in\Psi(x+p,A)$ and $p'>p$  (respectively, $p' <p$), there exists a trading strategy $\phi'\in\Psi (x+p',A)$ such that $V_t(x+p',\phi')>V_t(x+p,\phi)$ (respectively, $V_t(x+p',\phi')<V_t(x+p,\phi)$) for every $t\in [0,T]$.}
\end{assumption}

\begin{lemma}  \label{lem1.4}
If Assumption \ref{ass1.2} is satisfied, then the equality $\pinf^{s,i}(x_1,\cC^a)=\pinf^{a,i}(x_1,\cC^a)$ holds and thus
$\psup^{f,i}(x_1,\cC^a)=\pinf^{s,i}(x_1,\cC^a)=\pinf^{a,i}(x_1,\cC^a)$.
\end{lemma}

\proof
Let us first assume that  $\cH^{s,i}(x_1)\neq \emptyset$ or, equivalently, that $\pinf^{s,i}(x_1,\cC^a)<\infty $ (recall that $\inf \emptyset=\infty $). Since Assumption \ref{ass1.2} holds, it is clear that for an arbitrary $p\in\cH^{s,i}(x_1)$ and any $\varepsilon>0$,
there exists a strategy $\phi'\in\Psi (x_1+p+\varepsilon)$ such that condition (AO) is satisfied by the pair $(p+\varepsilon,\phi')$.
Hence $p+\varepsilon$ belongs to $\cH^{a,i}(x_1)$ and thus $p+\varepsilon \geq \pinf^{a,i}(x_1,\cC^a)$.
From the arbitrariness of $p \in  \cH^{s,i}(x_1)$ and $\varepsilon>0$, we infer that $ \pinf^{s,i}(x_1,\cC^a)\geq \pinf^{a,i}(x_1,\cC^a)$. From the discussion after Definition \ref{def1.12}, we get $\pinf^{s,i}(x_1,\cC^a)\leq\pinf^{a,i}(x_1,\cC^a)$ and thus $\pinf^{s,i}(x_1,\cC^a)=\pinf^{a,i}(x_1,\cC^a)$. Recalling that $ \psup^{f,i}(x_1,\cC^a)=\pinf^{a,i}(x_1,\cC^a)$ (see the comments after Definition \ref{def1.11}), we conclude that $\psup^{f,i}(x_1,\cC^a)=\pinf^{s,i}(x_1,\cC^a)=\pinf^{a,i}(x_1,\cC^a)$ if $\cH^{s,i}(x_1)\neq \emptyset$. Let us now assume that $\cH^{s,i}(x_1)=\emptyset$ so that $\cH^{a,i}(x_1)=\emptyset$ as well, since $\cH^{a,i}(x_1)\subseteq \cH^{s,i}(x_1)$. Then, on the one hand, $\pinf^{s,i}(x_1,\cC^a)=\pinf^{a,i}(x_1,\cC^a)=\infty $ and, on the other hand, $\psup^{f,i}(x_1,\cC^a)=\infty $, since $\cH^{f,i}(x_1)=\rr$ being the complement of $\cH^{a,i}(x_1)$. Hence the asserted equalities are satisfied in that case as well.
\endproof


\subsubsection{Holder's Superhedging Costs} \label{sec1.3.2}

Let us now examine the holder's strict superhedging costs.

\begin{definition} \label{def1.13}
{\rm The upper bound for holder's strict superhedging costs for $\cC^a$ is given by $\psup^{a,h}(x_2,\cC^a)$ $:=\sup \,\cH^{a,h}(x_2)$ where
\[
\cH^{a,h}(x_2):=\big\{ p\in\rr \,|\ \exists \, (\psi,\tau)\in\Psi (x_2-p,-A)\times\cT\!: (p,\psi ,\tau)\in\text{(AO$'$)} \big\}.
\]
If the equality $\psup^{a,h}(x_2,\cC^a)=\max \,\cH^{a,h}(x_2)$ holds, then $\psup^{a,h}(x_2,\cC^a)$ is denoted as $\wh{p}^{a,h}(x_2,\cC^a)$ and called the {\it holder's maximum strict superhedging cost} for $\cC^a$.}
\end{definition}

It is easily seen that $\cH^{a,h}(x_2)$ is the complement of $\cH^{f,h}(x_2)$. Hence we infer from Lemma \ref{lem1.3}
that the equality $\psup^{a,h}(x_2,\cC^a)=\pinf^{f,h}(x_2,\cC^a)$ is satisfied under Assumption \ref{ass1.1} and either
\begin{equation} \label{firstn}
\cH^{a,h}(x_2)=(-\infty,\psup^{a,c})\ \,\text{and}\ \, \cH^{f,h}(x_2)=[\,\breve{p}^{f,h},\infty)
\end{equation}
or
\begin{equation} \label{secondn}
\cH^{a,h}(x_2)=(-\infty,\wh{p}^{a,c}\,]\ \,\text{and}\ \, \cH^{f,h}(x_2)=(\pinf^{f,h},\infty).
\end{equation}

Similarly as for the issuer, we also introduce the concept of a {\it superhedging} strategy for the holder.

\begin{definition} \label{def1.14}
{\rm The upper bound for holder's superhedging costs for $\cC^a$ equals $\psup^{s,h}(x_2,\cC^a):=\sup \,\cH^{s,h}(x_2)$ where
\[
\cH^{s,h}(x_2):=\big\{p\in\rr \,|\ \exists \,(\psi,\tau)\in \Psi (x_2-p,-A)\times \cT\!: (p,\psi,\tau)\in \text{(SH$'$)}\big\}.
\]
If the equality $\psup^{s,h}(x_2,\cC^a)=\max\,\cH^{s,h}(x_2)$ holds, then $\psup^{s,h}(x_2,\cC^a)$ is denoted as $\wh{p}^{s,h}(x_2,\cC^a)$ and called the {\it holder's maximum superhedging cost} for $\cC^a$.}
\end{definition}

It is clear that $\cH^{a,h}(x_2) \subseteq \cH^{s,h}(x_2)$ and thus the inequality $\psup^{s,h}(x_2,\cC^a) \geq \psup^{a,h}(x_2,\cC^a)$ is valid. Furthermore, Assumption \ref{ass1.2} ensures that the equality holds.

\begin{lemma}  \label{lem1.5}
If Assumption \ref{ass1.2} is satisfied with $-A$, then the equality $\psup^{s,h}(x_2,\cC^a)=\psup^{a,h}(x_2,\cC^a)$ holds and thus
$\pinf^{f,h}(x_2,\cC^a)=\psup^{s,h}(x_2,\cC^a)=\psup^{a,h}(x_2,\cC^a)$.
\end{lemma}

\proof
Let us first assume that $\cH^{s,h}(x_2)\neq\emptyset$ or, equivalently, that $\psup^{s,h}(x_2,\cC^a)>-\infty$.  Under Assumption \ref{ass1.2}, for any $p\in\cH^{s,h}(x_2)$ and $\varepsilon>0$, there exists a pair $(\phi',\tau)\in\Psi (x_2-(p-\varepsilon))\times \cT $ such that $(p-\varepsilon,\phi',\tau)$ satisfy condition (AO$'$), that is, $p-\varepsilon\in\cH^{a,h}(x_2)$ and thus $p-\varepsilon\leq \psup^{a,h}(x_2,\cC^a)$. Since $p \in  \cH^{s,h}(x_2)$ and $\varepsilon>0$ are arbitrary, we conclude that $\psup^{s,h}(x_2,\cC^a)\leq \psup^{a,h}(x_2,\cC^a)$. The inequality $\psup^{s,h}(x_2,\cC^a)\leq \psup^{a,h}(x_2,\cC^a)$ is always satisfied and thus $\psup^{s,h}(x_2,\cC^a)=\psup^{a,h}(x_2,\cC^a)=\pinf^{f,h}(x_2,\cC^a)$ where the last equality is known to hold under Assumption \ref{ass1.1}. Let us now assume that $\cH^{s,h}(x_2)=\emptyset$. Then, on the one hand, $\psup^{s,h}(x_2,\cC^a)=\psup^{a,h}(x_2,\cC^a)=-\infty $ since
$\cH^{a,h}(x_2)\subseteq \cH^{s,h}(x_2)$ and, on the other hand, $\pinf^{f,h}(x_2,\cC^a)=-\infty $ since $\cH^{f,h}(x_2)$ is the complement of $\cH^{a,h}(x_2)$. Hence the asserted equalities are valid in that case as well.
\endproof

\subsection{Unilateral Acceptable Prices}   \label{sec1.4}

Our next goal is to analyze the following problem: under which assumptions a suitably defined {\it replication cost} of an American
contract is also its maximum (respectively, minimum) fair price for the issuer (respectively, the holder). The answer to this question
will lead us to the crucially important concept of unilateral acceptable prices computed by the counterparties.

\subsubsection{Issuer's Acceptable Price}   \label{sec1.4.1}

We will now study the concept of replication of the contract $\cC^a$ by the issuer. We work hereafter under Assumption \ref{ass1.2} and thus, in view of Lemma \ref{lem1.4}, we have that
\begin{equation} \label{equah}
\psup^{f,i}(x_1,\cC^a)=\pinf^{s,i}(x_1,\cC^a)=\pinf^{a,i}(x_1,\cC^a).
\end{equation}

\begin{definition} \label{def1.15}
{\rm The lower bound for issuer's replication costs for $\cC^a$ is given by $\pinf^{r,i}(x_1,\cC^a)$ \\ $:=\inf \,\cH^{r,i}(x_1)$ where
\begin{align*}
\cH^{r,i}(x_1):=\big\{p\in\rr\,|\,\exists \,(\phi,\tau)\in\Psi(x_1+p,A)\times\cT\!:\text{$(p,\phi)\in $ (SH) \& } (p,\phi,\tau)\in\text{(BE)}\big\}.
\end{align*}
If the equality $\pinf^{r,i} (x_1,\cC^a)=\min \,\cH^{r,i}(x_1)$ holds, then $\pinf^{r,i} (x_1,\cC^a)$ is denoted as $\breve{p}^{r,i}(x_1,\cC^a)$ and called the {\it issuer's minimum replication cost} for $\cC^a$.}
\end{definition}

Note that in Definition \ref{def1.15} we focus on a particular issuer's superhedging strategy for which a break-even time exists
and we do not impose any restrictions on wealth processes of other trading strategies available to the issuer. Hence, in principle,
for $p\in \cH^{r,i}(x_1)$, it may happen that there exists another pair, say $(p,\psi )$, which is an issuer's strict superhedging strategy. This would mean that the issuer's replication cost is not an issuer's fair price for $\cC^a$. To eliminate this shortcoming of Definition \ref{def1.15}, in the next definition we impose, in addition, the no-arbitrage condition (NA) on all issuer's trading strategies associated with $p$.

\begin{definition} \label{def1.16}
{\rm The lower bound for issuer's fair replication costs for $\cC^a$ is given by $\pinf^{f,r,i}(x_1,\cC^a):=\inf \,\cH^{f,r,i}(x_1)$ where
\begin{align*}
\cH^{f,r,i}(x_1):=& \big\{ p\in\rr\,|\ \exists \, (\phi ,\tau)\in\Psi (x_1+p,A)\times\cT\!:
(p,\phi)\in\text{(SH) \& } (p,\varphi,\tau)\in\text{(BE)}; \\
&\quad  \quad \quad \ \ \,  \forall \,\phi'\in\Psi (x_1+p,A) \, \exists \,\tau'\in\cT\!: \, (p,\phi',\tau')\in\text{(NA)} \big\}.
\end{align*}
If the equality $\pinf^{f,r,i}(x_1,\cC^a)=\min \,\cH^{f,r,i}(x_1)$ holds, then $\pinf^{f,r,i}(x_1,\cC^a)$ is denoted as $\breve{p}^{f,r,i}(x_1,\cC^a)$ and called the {\it issuer's minimum fair replication cost} for $\cC^a$.}
\end{definition}

\begin{definition} \label{def1.17}
{\rm If the set $\cH^{f,r,i}(x_1)$ is a singleton, then its unique element is denoted as $p^i(x_1,\cC^a)$ and called the {\it issuer's acceptable price}.}
\end{definition}

Obviously, if $p^i(x_1,\cC^a)$ is well defined, then it coincides with $\breve{p}^{f,r,i}(x_1,\cC^a)$.
Let us examine some useful relationships between the conditions introduced in Definition \ref{def1.1}.

\begin{lemma} \label{lem1.6}
(i) If $(p,\phi)$ satisfy (SH) and $(p,\phi,\tau)$ satisfy (NA), then (BE) is valid for $(p,\phi,\tau)$. \hfill \break
(ii) If $(p,\phi)$ fail to satisfy (SH), then there exists $\tau\in\cT $ such that $(p,\phi,\tau)$ satisfy (NA). \hfill \break
(iii) The following conditions are equivalent: \hfill \break
(a) for any $(p,\phi)$ satisfying (SH) there exists $\tau\in\cT $ such that $(p,\phi,\tau)$ satisfy (BE),\hfill \break
(b) for any $(p,\phi)$, there exists $\tau\in\cT $ such that $(p,\phi,\tau)$ satisfy (NA).
\end{lemma}

\proof
\noindent (i) The statement is an almost immediate consequence of definitions of conditions (SH), (NA) and (BE). Indeed, suppose that a pair $(p,\phi)$ satisfies (SH) and a triplet $(p,\phi,\tau)$ complies with (NA). Clearly, the inequality $\P (V_{\tau}(x_1+p,\phi)+\Xhh_{\tau}< \Vben_{\tau}(x_1))>0$ cannot hold and thus $(p,\phi,\tau)$ must satisfy (BE).

\noindent (ii)  We know that if a pair $(p,\phi)$ satisfies (AO), then it fulfills (SH) and thus if $(p,\phi)$ fails to satisfy (SH), then condition (AO) is not met. Recalling that, in essence, the class of pairs $(p,\phi)$ satisfying (AO) is the complement of (NA), we conclude that there exists a stopping time $\tau\in\cT $ such that the triplet $(p,\phi,\tau)$ complies with (NA).

\noindent (iii) We first show that (a) implies (b).  We know from (ii) that if a pair $(p,\phi)$ does not satisfy condition (SH), then
there exists $\tau\in\cT $ such that $(p,\phi,\tau)$ complies with (NA). If $(p,\phi)$ satisfies (SH)
then, from condition (a), there exists $\tau\in\cT $ such that the triplet $(p,\phi,\tau)$ fulfills (BE), and thus it also satisfies (NA). Let us now assume that condition (b) is met. Then, from part (i), condition (a) holds as well.
\endproof

In view of \eqref{equah}, the proof of the following lemma is obvious and thus it is omitted.

\begin{lemma} \label{lem1.7}
We have that $\cH^{s,i}(x_1) \supseteq \cH^{r,i}(x_1)\supseteq \cH^{f,r,i}(x_1)=\cH^{f,i}(x_1) \cap \cH^{r,i}(x_1)$ and thus
\begin{equation} \label{xprices1}
\psup^{f,i}=\pinf^{s,i} \leq \pinf^{r,i} \leq \pinf^{f,r,i}.
\end{equation}
\end{lemma}

In the next result, we study the basic properties of issuer's costs.

\begin{proposition} \label{pro1.1}
(i) If $\cH^{f,r,i}(x_1)\ne\emptyset $, then it is a singleton and the issuer's acceptable price $p^i = p^i(x_1,\cC^a)$ satisfies
\begin{equation} \label{bg11}
-\infty< p^i = \wh{p}^{f,i}=\breve{p}^{r,i}=\breve{p}^{s,i}<+\infty .
\end{equation}
(ii) If $\cH^{r,i}(x_1)\ne\emptyset $, then $\psup^{f,i}=\pinf^{s,i}\leq \pinf^{r,i}<\infty$. \\
(iii) If $\cH^{r,i}(x_1)=\emptyset $, then $\psup^{f,i}=\pinf^{s,i}\leq \pinf^{r,i}=\infty$.
\end{proposition}

\proof We start by proving (i). If $\cH^{f,r,i}(x_1) \ne \emptyset$, then from the inclusion $\cH^{f,r,i}(x_1)\subseteq \cH^{f,i}(x_1)$
we obtain the inequality $\pinf^{f,r,i} \leq \psup^{f,i}$ and thus, in view of \eqref{xprices1}, we have $\psup^{f,i}=\pinf^{s,i}=\pinf^{r,i}=\pinf^{f,r,i}$. Moreover, $\cH^{f,i}(x_1)\ne \emptyset$ and $\cH^{s,i}(x_1)\ne \emptyset$ and thus $\psup^{f,i}>-\infty $ and $\pinf^{s,i}<\infty $. We conclude that $-\infty <\psup^{f,i}=\pinf^{r,i}=\pinf^{f,r,i}=\pinf^{s,i}<\infty $. From the inclusion $\cH^{f,r,i}(x_1)\subseteq \cH^{f,i}(x_1)$ and the equalities $\sup \,\cH^{f,i}(x_1)=\psup^{f,i}=\pinf^{f,r,i}=\inf \,\cH^{f,r,i}(x_1)$, we deduce that for arbitrary $p_1,p_2\in\cH^{f,r,i}(x_1)$ and $p_3\in \cH^{f,i}(x_1)$ we have $p_1=p_2\ge p_3$ and thus $\cH^{f,r,i}(x_1)$ is a singleton and its unique element is not less than any element of $\cH^{f,i}(x_1)$. Consequently, $\wh{p}^{f,i}$ and $p^i$ are well defined and satisfy $-\infty<\wh{p}^{f,i}=p^i<+\infty$. Furthermore, $\cH^{f,r,i}(x_1) \subseteq  \cH^{r,i}(x_1)$ and thus the equality $p^i=\pinf^{r,i}$ implies that $\breve{p}^{r,i}$ exists as well and coincides with $p^i$. We conclude that \eqref{bg11} is valid. This means, in particular, that $\cH^{f,i}(x_1)=(-\infty ,\wh{p}^{f,i}]=(-\infty ,\breve{p}^{r,i}]=(-\infty ,p^i]$. Parts (ii) and (iii) are easy consequences of \eqref{xprices1}.
\endproof

\begin{remark} \label{rem1.2}
{\rm Part (i) in Proposition \ref{pro1.1} shows that if there exists a number $p$ for which an issuer's fair replicating
strategy exists and $p$ is the unique issuer's acceptable price, his maximum fair price,
minimum replication cost, and the infimum of (strict) superhedging costs (but, obviously, it is not a strict superhedging cost). We thus deal here with a highly desirable situation but, sadly, it is not easy to check whether a number $p$ with above-mentioned properties exists.  As expected, to overcome this difficulty we will impose additional assumptions on wealth processes of trading strategies. Recall that in the complete linear market, one can show that $\cH^{r,i}(x_1)\ne\emptyset$ but, to the best of our knowledge, the properties of the set $\cH^{f,r,i}(x_1)$ have not been studied in the existing literature.
} \end{remark}

\subsubsection{Holder's Acceptable Price} \label{sec1.4.2}

We postulate that Assumption \ref{ass1.2} is satisfied with $-A$ and thus, in view of Lemma \ref{lem1.5}, we have that
\begin{equation} \label{equac}
\pinf^{f,h}(x_2,\cC^a)=\psup^{s,h}(x_2,\cC^a)=\psup^{a,h}(x_2,\cC^a).
\end{equation}
The notion of the holder's replication cost introduced is through the following definition.

\begin{definition} \label{def1.18}
{\rm The upper bound for holder's replication costs for $\cC^a$ is given by  $\psup^{r,h}(x_2,\cC^a)$ \\ $:=\sup \,\cH^{r,h}(x_2)$ where
\begin{equation*}
\cH^{r,h}(x_2)=\big\{ p\in\rr \,|\   \exists \, (\psi,\tau)\in\Psi (x_2-p,-A)\times\cT\!: (p,\psi ,\tau)\in\text{(BE$'$)} \big\}.
\end{equation*}
Equivalently,
\begin{equation*} 
\psup^{r,h}(x_2,\cC^a):=-\inf \big\{ q\in\rr \,|\  \exists \, (\psi ,\tau)\in\Psi (x_2+q,-A)\times\cT\!: V_{\tau}(x_2+q,\psi )-\Xhh_{\tau}=\Vben_{\tau}(x_2) \big\}.
\end{equation*}
If the equality $\psup^{r,h}(x_2,\cC^a)=\max \,\cH^{r,h}(x_2)$ holds, then $\psup^{r,h}(x_2,\cC^a)$ is denoted as $\wh{p}^{r,h}(x_2,\cC^a)$ and called the {\it holder's maximum replication cost} for $\cC^a$.}
\end{definition}

To establish the existence of the holder's acceptable price, we will use the concept of the holder's fair replication costs.

\begin{definition} \label{def1.19}
{\rm The upper bound for holder's fair replication costs for $\cC^a$ is given by $ \psup^{f,r,h}(x_2,\cC^a)=\sup \,\cH^{f,r,h}(x_2)$ where
\begin{align*}
\cH^{f,r,h}(x_2):=\big\{ p\in\rr \,|\ & \exists \, (\psi,\tau)\in\Psi (x_2-p,-A)\times\cT\!:
\text{$(p,\psi,\tau) \in$ (BE$'$)}  \ \&  \\ &  \forall \, (\psi',\tau' )\in\Psi (x_2-p,-A)\times\cT\!:  \text{$(p,\psi',\tau' ) \in $ (NA$'$)} \big\}.
\end{align*}
If $\psup^{f,r,h}(x_2,\cC^a)=\max \,\cH^{f,r,h}(x_2)$, then $\psup^{f,r,h}(x_2,\cC^a)$ is denoted as $\wh{p}^{f,r,h}(x_2,\cC^a)$ and called the {\it  holder's maximum fair replication cost} for $\cC^a$.}
\end{definition}

\begin{definition} \label{def1.20}
{\rm If the set $\cH^{f,r,h}(x_2)$ is a singleton, then its unique element is denoted as $p^h(x_2,\cC^a)$ and called the {\it holder's acceptable price}.}
\end{definition}

Notice that if  $p^h(x_2,\cC^a)$ is well defined, then it is equal to $\wh{p}^{f,r,h}(x_2,\cC^a)$. The following lemma is an easy consequence of Definition \ref{def1.7}.

\begin{lemma} \label{lem1.8}
(i) If $(p,\psi,\tau)$ satisfy (BE$'$), then they satisfy (SH$'$). \\
(ii) If $(p,\psi,\tau)$ satisfy (SH$'$) and (NA$'$), then they satisfy (BE$'$). \\
(iii) If $(p,\psi,\tau)$ do not satisfy (SH$'$), then they satisfy (NA$'$).
\end{lemma}

We infer from Lemma \ref{lem1.8} that $\cH^{s,h}(x_2) \supseteq  \cH^{r,h}(x_2)\supseteq \cH^{f,r,h}(x_2)=\cH^{r,h}(x_2) \cap \cH^{f,h}(x_2)$ and thus, in view of \eqref{equac}, we have
\begin{equation} \label{cequac}
\pinf^{f,h}=\psup^{s,h} \geq \psup^{r,h} \geq \psup^{f,r,h}.
\end{equation}

\begin{proposition} \label{pro1.2}
(i) If $\cH^{f,r,h}(x_2)\ne\emptyset $, then it is a singleton and the holder's acceptable price $p^h = p^h(x_2,\cC^a)$ satisfies
\begin{equation} \label{bg22}
-\infty< p^h = \breve{p}^{f,h}=\wh{p}^{r,h}=\wh{p}^{s,h}<+\infty .
\end{equation}
(ii) If $\cH^{r,h}(x_2) \ne \emptyset $, then $-\infty<\psup^{r,h} \leq\pinf^{f,h}=\psup^{s,h}$.\\
(iii) If $\cH^{r,h}(x_2)=\emptyset $, then $-\infty=\psup^{r,h}\leq\pinf^{f,h}=\psup^{s,h}$.
\end{proposition}

\proof
It is clear that $\cH^{f,r,h}(x_2)\subseteq  \cH^{f,h}(x_2)$ and thus if $\cH^{f,r,h}(x_2) \ne \emptyset $, then $\psup^{f,r,h}<+\infty$ and $\psup^{f,r,h} \geq \pinf^{f,h}$. Furthermore, the sets $\cH^{s,h}(x_2)$ and $\cH^{r,h}(x_2)$ are nonempty and thus we infer from \eqref{cequac} that $- \infty < \pinf^{f,h}=\psup^{r,h}=\psup^{f,r,h}=\psup^{s,h} < \infty $. From the inclusion $\cH^{f,r,h}(x_2)\subseteq \cH^{f,h}(x_2)$ and the equalities $\inf \,\cH^{f,h}(x_2)=\pinf^{f,h}=\psup^{f,r,h}=\sup \,\cH^{f,r,h}(x_2)$,  we see that for arbitrary $p_1,p_2\in\cH^{f,r,h}(x_2)$ and $p_3\in \cH^{f,h}(x_2)$ we have $p_1=p_2\leq p_3$ and thus $\cH^{f,r,h}(x_2)$ is a singleton and its unique element is less than any element of $\cH^{f,h}(x_2)$. Consequently, $\breve{p}^{f,h}$ and $p^h$ are well defined and they satisfy $-\infty< p^h =\breve{p}^{f,h}<+\infty$. Furthermore, $\cH^{f,r,h}(x_2) \subseteq  \cH^{r,h}(x_2)$ and thus the equality $p^h=\psup^{r,h}$ implies that $\wh{p}^{r,h}$ exists and coincides with $p^h$. We conclude that \eqref{bg22} is valid. The proofs of statements (ii) and (iii) are straightforward in view of \eqref{cequac}.
\endproof


\section{Unilateral Pricing Through Reflected BSDEs}  \label{sec4}

The goal of this section is to re-examine and extend a BSDE approach to the valuation of American options in nonlinear market, which was initiated in the paper by Dumitrescu et al.~\cite{DQS2017}. Our main goal is to show that unilateral acceptable prices for an American contract $\cC^a$ can be characterized in terms of solutions to reflected BSDEs driven by a multi-dimensional continuous semimartingale $S$.  In this section, we postulate that the wealth process $V=V(y,\phi,A)$ satisfies the forward equation
\begin{equation} \label{ueq3x}
V_t=y-\int_0^t g(u,V_u,\xi_u)\,du+\int_0^t\xi_u\,dS_u+A_t,
\end{equation}
where $y \in \rr$ represents the initial wealth at time 0 of a given trading strategy $\xi$ (recall that $A_0=0$). By applying Lemma \ref{lem:1} with $y_1=x+p<x+p'=y_2,\,f_1=f_2=g$ and $z=\xi$, it is easy to check that Assumption \ref{ass1.2} is met when the dynamics of the wealth process are given by the SDE \eqref{ueq3x} and are uniquely specified by the initial value $y$ and the process $\xi$ (for instance, when $g=g(t,v,z)$ is Lipschitz continuous with respect to $v$).

\subsection{Dynamics of the Wealth Process} \label{sec4.1}

Let us first describe more explicitly the main features of the mechanism of nonlinear trading, which underpins the wealth dynamics given by \eqref{ueq3x}. We start by introducing the notation for traded assets, that is, cash accounts, risky assets, and funding accounts associated with risky assets. It should be stressed, however, that our further developments will not depend on the choice of a particular model for primary assets and trading arrangements and thus our general results are capable of covering a broad spectrum of market models.

Let $\cS=(S^1,\ldots,S^d)$ stand for the collection of prices of a family of $d$ risky assets where the processes $S^1,\dots,S^d$ are continuous semimartingales. Continuous processes of finite variation, denoted as $B^{0,l}$ and $B^{0,b}$, represent the {\it lending} {\it borrowing} unsecured cash accounts, respectively. For every $j=1,2,\ldots,d$, we denote by $B^{j,l}$ (respectively, $B^{j,b}$) the {\it lending} (respectively, {\it borrowing}) {\it funding account} associated with the $i$th risky asset, and also assumed to be continuous processes of finite variation. The financial interpretation of these accounts varies from case to case (for more details, see \cite{BCR2018,BR2015}).  Let us denote by $\cB$ the collection of all cash and funding accounts available to a trader. For simplicity of presentation, we maintain our assumption that the issuer and the holder have identical market conditions but it is clear that this assumption is not relevant for our further developments and thus it can be easily relaxed.
A {\it trading strategy} is an ${\mathbb R}^{3d+2}$-valued, $\gg$-adapted process $\phi=( \xi^1,\ldots,\xi^{d}; \psi^{0,l},\psi^{0,b},\dots ,\psi^{d,l}, \psi^{d,b})$ where the components represent all outstanding positions in the risky assets $S^j,\, j=1,2, \dots,d$, cash accounts $B^{0,l},\, B^{0,b}$, and funding accounts $B^{j,l}, B^{j,b} ,\, j=1,2, \dots,d$ for risky assets.

\begin{definition}\label{self financing} {\rm
We say that a trading strategy $(y,\phi)$ is {\it self-financing} for $\cC^a$ and we write $\phi \in \Psi (y,A)$ if the {\it wealth process} $V(y,\phi,A)$, which is given by
\begin{equation*} 
V_t(y,\phi,A)=\sum_{j=1}^{d}\xi^j_t S^j_t+\sum_{j=0}^{d}\psi^{j,l}_t B^{j,l}_t+\sum_{j=0}^{d}\psi^{j,b}_t B^{j,b}_t ,
\end{equation*}
satisfies, for every $t \in [0,T]$,
\begin{equation*} 
V_t(y,\phi,A)=y+\sum_{j=1}^{d}\int_0^t\xi^j_u\,dS^j_u+\sum_{j=0}^{d}\int_0^t\psi^{j,l}_u\,dB^{j,l}_u+\sum_{j=0}^{d}\int_0^t\psi^{j,b}_u\,dB^{j,b}_u + A_t
\end{equation*}
subject to additional constraints imposed on the components of $\phi $. In particular, we postulate that $\psi^{j,l}_t \geq 0,\, \psi^{j,b}_t \leq 0$ and $\psi^{j,l}_t \psi^{j,b}_t=0$ for all $j=0,1, \dots ,d$ and $t \in [0,T]$.
} \end{definition}

Due to additional trading constraints, which depend on the particular trading mechanism, the choice of an initial value $y$ and a process $\xi $ is known to uniquely specify the wealth process of a self-financing strategy $\phi\in\Psi (y,A)$. In addition, one needs also to introduce some form of {\it admissibility} of trading strategies and to postulate that the market model $\cM=(\cB ,\cS ,\Psi (A))$ where the class $\Psi (A) = \cup_{y\in \rr}\Psi (y,A)$ of all admissible trading strategies is arbitrage-free in a suitable sense, for instance, the market model $\cM$ can be assumed to be {\it regular}, in the sense of Bielecki et al.~\cite{BCR2018}.

Since the arbitrage-free feature of a nonlinear market was studied in El Karoui and Quenez~\cite{EQ1997} and, more recently, Bielecki et al.~\cite{BCR2018,BR2015}, we do not elaborate on this issue here and we refer the interested reader to these works.
It is important to notice that, due to the trading constraints, differential funding costs and possibly also some additional adjustment processes, which are not explicitly stated in Definition \ref{self financing}, the dynamics of the wealth process are nonlinear, in general. We refer the reader to Bielecki et al.~\cite{BCR2018,BR2015} for more encompassing versions of the self-financing property of a trading strategy (see, for instance, Definition 4.5 in \cite{BR2015} or Definition 2.2 in \cite{BCR2018}) and to Nie and Rutkowski~\cite{NR1,NR2,NR3,NR4} for explicit examples of nonlinear models with trading constraints and adjustments (in particular, the collateralization of contracts). We only observe that each particular trading arrangement gives rise to an explicit mapping $g$
in the dynamics \eqref{ueq3x} of the wealth, which is sometimes quite complex, but at the same time usually fairly regular.

The following elementary lemma addresses the issue of the (strict) monotonicity of the wealth process driven by \eqref{ueq3x}. Note that since the process $z$ is assumed to be given in the statement of Lemma \ref{lem:1}, we may interpret the SDE \eqref{eq:forwardSDE1} as a deterministic integral equation, which holds for almost all $\omega \in \Omega$.


\begin{lemma} \label{lem:1}
Let $g_j:\Omega\times[0,T]\times\rr\times\rr^d\rightarrow\rr,\,j=1,2$ be $\cP\otimes\cB(\rr)\otimes\cB(\rr^d)/\cB(\rr)$-measurable. Consider the SDEs for $j=1,2$
\begin{equation} \label{eq:forwardSDE1}
v^j_t=y_j-\int_0^t g_j(u,v^j_u,z_u)\,du+k^j_t+\int_0^t z_u\,dS_u + H_t ,
\end{equation}
where $z$ is an $\rr^d$-valued, $\gg$-progressively measurable stochastic process, the $\gg$-adapted process $k=k^1-k^2$ is nondecreasing with $k_0=0$, $S$ is an $\rr^d$-valued semimartingale, and $H$ is a c\`adl\`ag, $\gg$-adapted process.  Assume that \eqref{eq:forwardSDE1} has a unique solution $v^j$ for $j=1,2$. If $g_1(t,v^2_t,z_t)\leq g_2(t,v^2_t,z_t),\,dt\otimes d\P$-a.e. and $y_1\geq y_2$ (respectively, $y_1>y_2$), then $v^1_t\geq v^2_t$ (respectively, $v^1_t>v^2_t$) for all $t \in [0,T]$.
\end{lemma}

\proof
If we set $\bar{v}_t:=v^1_t-v^2_t$ for all $t\in [0,T]$, then $\bar{v}$ is easily seen to satisfy
\begin{align*}
\bar{v}_t&=y_1-y_2+\int_0^t\big(g_2(u,v^2_u,z_u)-g_1(u,v^1_u,z_u)\big)\,du+k_t  \\
&=y_1-y_2+\int_0^t\big(\lambda_u \bar{v}_u+g_2(u,v^2_u,z_u)-g_1(u,v^2_u,z_u)\big)\,du+k_t
\end{align*}
where $\lambda_u:=\big(g_1(u,v^2_u,z_u)-g_1(u,v^1_u,z_u)\big)(\bar{v}_u)^{-1}\I_{\{v^1_u\ne v^2_u\}}$. Thus
\[
d\Big(e^{-\int_0^t\lambda_u\,du}\bar{v}_t\Big)=e^{-\int_0^t\lambda_u \,du}\big((g_2(t,v^2_t,z_t)-g_1(t,v^2_t,z_t))\,dt+dk_t\big).
\]
After integrating both sides from $0$ to $t$, and noticing that the process $k=k^1-k^2$ is nondecreasing and $e^{-\int_0^t\lambda_u \,du}>0$, we obtain the inequality
\[
\bar{v}_t\ge\bar{v}_0\,e^{\int_0^t\lambda_u \,du}+\int_0^t e^{\int_s^t\lambda_u \,du}\big(g_2(s,v^2_s,z_s)-g_1(s,v^2_s,z_s)\big)\,ds
\]
from which the assertion of the lemma follows.
\endproof

\subsection{Comparison Properties of Nonlinear Evaluations} \label{sec4.1xx}

We henceforth postulate that the wealth process $V=V(y,\phi,A)$ is governed by the SDE \eqref{ueq3x} where $\xi=(\xi^1, \dots,\xi^d)$ is given and the mapping $g$ satisfies some additional assumptions. This will allow us to refer to a large body of existing literature on the theory of BSDEs and, in particular, to exploit the link between reflected BSDEs and solutions to nonlinear optimal stopping problems (see, for instance, Cvitani\'c and Karatzas~\cite{CK1996}, Dumitrescu et al.~\cite{DQS2016}, El Karoui et al.~\cite{EQK1997}, Grigorova and Quenez~\cite{GQ2017}, Grigorova et al.~\cite{GIOOQ2017,GIOQ2017}, and Quenez and Sulem~\cite{QS2014}). Our goal in this section is to examine the valuation and hedging of American contracts for the special case where the wealth dynamics are driven by \eqref{ueq3x}. To keep our presentation concise and encompassing several alternative nonlinear market models, we will directly postulate that the associated BSDEs enjoy desirable properties, such as: the existence, uniqueness, and the strict comparison property for solutions to BSDEs, which are known to hold under various circumstances. As was already mentioned, a particular instance of a market model given by \eqref{ueq3x} has been studied in a recent paper by Dumitrescu et al.~\cite{DQS2017}.

We will use the standard terminology related to nonlinear evaluations generated by solutions to BSDEs (see, e.g., Chapter 3 in Peng~\cite{P2004a} or Section 4 in Peng \cite{P2004b}). Consider the following BSDE on $[0,s]$
\begin{equation}\label{xBSDEU}
Y_t =\zeta_s +\int_t^s g(u,Y_u,Z_u)\,du-\int_t^s Z_u\,dM_u-(\hHh_s-\hHh_t),
\end{equation}
where $\zeta_s \in L^2(\cG_s)$, $M$ is a $d$-dimensional martingale, the process $\hHh$ is real-valued and $\gg$-adapted,
and the generator $g: \Omega\times[0,T]\times\rr\times\rr^d \rightarrow \rr$ is $\cP \otimes \cB (\rr )\otimes\cB (\rr^d)/\cB(\rr)$-measurable where $\cP$ is the $\sigma$-field of predictable sets on $\Omega\times[0,T]$. Assume that the BSDE \eqref{xBSDEU} has a unique solution $(Y,Z)$ in a suitable space of stochastic processes (see, e.g., \cite{CFS2008,NR2}).
For every $0\leq t\leq s \leq T$ and $\zeta_s \in L^2(\cG_s)$, we denote $\cEgH_{t,s}(\zeta_s )= Y_t$ where $(Y,Z)$ solves the BSDE \eqref{xBSDEU} with $Y_s=\zeta_s$. Then the system of operators $\cEgH_{t,s}: L^2(\cG_s) \to L^2(\cG_t)$ is called the {\it $\cEgH$-evaluation}.  It is worth noting that a deterministic dates $t \leq s$ appearing in the BSDE \eqref{xBSDEU} can be replaced by arbitrary $\gg$-stopping times $\tau \leq \sigma $ from $\cT$ and thus the notion of the $\cEgH$-evaluation can be extended to stopping times $\cEgH_{\tau ,\sigma }: L^2(\cG_\sigma ) \to L^2(\cG_\tau )$. The concept of the (strict) comparison property is of great importance in the theory of BSDEs and nonlinear evaluations.

\begin{definition} \label{xdefmong} {\rm
We say that the {\it comparison property} of $\cEgH$ holds if for every stopping time $\tau\in\cT$ and random variables $\zeta^1_{\tau},\zeta^2_{\tau} \in L^2(\cG_{\tau})$, the following property is valid: if $\zeta^1_{\tau} \geq \zeta^2_{\tau}$ then $\cEgH_{0,\tau }(\zeta^1_{\tau})\geq \cEgH_{0,\tau }(\zeta^2_{\tau})$. We say that the {\it strict comparison property} of $\cEgH$ holds if for every $\tau\in\cT$ and $\zeta^1_{\tau},\zeta^2_{\tau} \in L^2(\cG_{\tau})$ if $\zeta^1_{\tau}\geq\zeta^2_{\tau}$ and $\zeta^1_{\tau}\ne\zeta^2_{\tau}$ then $\cEgH_{0,\tau}(\zeta^1_{\tau})>\cEgH_{0,\tau}(\zeta^2_{\tau})$.}
\end{definition}

In view of its financial interpretation, the nonlinear evaluation $\cE^{g,A}$ associated with the BSDE
\begin{equation} \label{nhBSDE}
Y_t=\zeta_s+\int_t^s g(u,Y_u,Z_u)\,du-\int_t^s Z_u\,dS_u-(A_s-A_t)
\end{equation}
is henceforth denoted by $\cEgi$ and called the {\it issuer's $g$-evaluation}. In Section \ref{sec4.3}, we address the issuer's pricing problem and we work under the following standing assumption.

\begin{assumption} \label{assnBSDE} {\rm We postulate that:\\
(i) the wealth process $V=V(y,\phi,A)$ of a trading strategy $\phi \in \Psi (y,A)$ satisfies \eqref{ueq3x},\\
(ii) for any fixed $y \in \rr$ and any process $\xi$ such that the stochastic integral in \eqref{ueq3x} is well defined,
the SDE \eqref{ueq3x} has a unique strong solution, \\
(iii) the strict monotonicity property holds for the wealth $V(y,\phi,A)$ (see Assumption \ref{ass1.2}),\\
(iv) for every $(s,\zeta_s)\in [0,T]\times L^2(\cG_s)$ the BSDE \eqref{nhBSDE} has a unique solution $(Y,Z)$ on $[0,s]$}.
\end{assumption}

Note that in Section \ref{sec4.4}, where we study the holder's pricing problem, we will use a modified version of Assumption \ref{assnBSDE} in which the process $A$ is replaced by $-A$.


\begin{remark} {\rm
In view of Lemma \ref{lem:1} and condition (ii) in Assumption \ref{assnBSDE}, condition (iii) in Assumption \ref{assnBSDE} is not restrictive, since it is satisfied for every generator $g$. For explicit assumptions about $g$ ensuring that the BSDE \eqref{nhBSDE}, where $S$ is a multi-dimensional, continuous, square-integrable martingale enjoying the predictable representation property, has a unique solution and the strict comparison property of the issuer's $g$-evaluation $\cEgi$ holds, see Theorems 3.2 and 3.3 in Nie and Rutkowski~\cite{NR2}.}
\end{remark}

\subsection{Issuer's Acceptable Price via RBSDE} \label{sec4.3}

We denote by $X:=\Vben(x_1)-\Xhh$ the {\it issuer's relative reward} and we assume that the process $X$ is square-integrable. Then,
by Assumption \ref{assnBSDE}(iv), the following BSDE on $[0,T]$
\begin{equation} \label{nBSDE}
Y_t =X_T +\int_t^T g(u,Y_u,Z_u)\,du-\int_t^T Z_u\,dS_u-(A_T-A_t)
\end{equation}
has a unique solution $(Y,Z)$ in a suitable space of stochastic processes.
To study the issuer's pricing problem for an American contract $\cC^a$, we make the following assumption
(see, e.g., Definition 2.4 in Quenez and Sulem~\cite{QS2014}).

\begin{assumption} \label{assBSDEi}
{\rm The reflected BSDE with the lower obstacle $X$
\begin{equation} \label{iRBSDE}
\left\{ \begin{array} [c]{ll}
dY_t=-g(t,Y_t,Z_t)\,dt+Z_t\,dS_t+dA_t-d\Kpp_t,\medskip\\ Y_T=X_T,\quad Y_t \geq X_t,\quad \int_0^T (Y_t-X_t)\,d\Kpp^c_t=0,
\quad \Delta \Kpp^d_t=-\Delta (Y_t-A_t) \I_{\{ Y_{t-}= X_{t-}\}},
\end{array} \right.
\end{equation}
has a unique solution $(Y,Z,\Kpp)$ where $K$ is a $\gg$-predictable, c\`adl\`ag, nondecreasing process such that $K_0=0$ and $K=K^c+K^d$ is its unique decomposition into continuous and jump components.}
\end{assumption}

Reflected BSDEs were introduced in seminal papers by El Karoui et al.~\cite{EPAQ1997,EQK1997} where they were applied to solutions of optimal stopping problems and pricing of American options. They were subsequently studied by several authors who dealt with various frameworks, to mention a few: Aazizi and Ouknine~\cite{AO2016}, Baadi and Ouknine~\cite{BB2017}, Cr\'epey and Matoussi \cite{CM2008}, Essaky~\cite{ES2008}, Hamad\`ene~\cite{HA2002}, Hamad\`ene and Ouknine~\cite{HO2016}, Grigorova et al.~\cite{GIOOQ2017,GIOQ2017}, Klimsiak~\cite{K2012,K2015}, Klimsiak et al.~\cite{KRS2016} and Quenez and Sulem~\cite{QS2014}. They have shown that Assumption \ref{assBSDEi} is met in many instances of our interest but, due to space limitations, we are not going to quote any particular result here.

It is worth stressing that we deliberately do not specify particular spaces
of stochastic processes in which the components $Y$ and $Z$ are searched for, since our further results do not depend on
the choice of these spaces. Only the properties of the  process $K$ in a unique solution $(Y,Z,K)$ to the issuer's RBSDE \eqref{iRBSDE} (and, by the same token, of the process $k$ in a unique solution $(y,z,k)$ to the holder's RBSDE  \eqref{hRBSDE}) are essential in what follows and thus they are stated explicitly and analyzed in some detail.

\begin{definition} \label{def2vxa} {\rm
We say that $\vsup^i(\cC^a)\in\rr$ is the {\it value} of the issuer's optimal stopping problem for $\cC^a$ if
\begin{equation*}
\vsup^i(\cC^a)=\sup_{\tau\in\cT}\,\cEgi_{0,\tau }(X_{\tau}).
\end{equation*}
A stopping time $\taus\in\cT$ is called a {\it solution} to the issuer's optimal stopping problem if $\vsup^i(\cC^a)=\wh{v}^{i}(\cC^a)$ where}
\begin{equation}  \label{eq2vb}
\wh{v}^{i}(\cC^a):=\cEgi_{0,\taus }(X_{\taus})=\max_{\tau\in\cT}\,\cEgi_{0,\tau }(X_{\tau}).
\end{equation}
\end{definition}

\begin{assumption} \label{assBSDEm1}
{\rm The value $\vsup^i(\cC^a)$ to the issuer's optimal stopping problem exists and satisfies $\vsup^i(\cC^a)=Y_0$.}
\end{assumption}

\begin{assumption} \label{assBSDEm2}
{\rm The stopping time $\tau^i:=\inf\,\{t\in [0,T]\,|\,Y_t=X_t\}$ is a (not necessarily unique) solution to the issuer's optimal stopping problem so that $\wh{v}^i(\cC^a)=\cEgi_{0,\tau^i }(X_{\tau^i })$.}
\end{assumption}

\begin{remark} {\rm Various results pertaining to Assumptions \ref{assBSDEi}--\ref{assBSDEm2} were obtained under alternative assumptions on the generator $g$ and the processes $X$ and $S$ in \eqref{iRBSDE} by, among others, El Karoui et al.~\cite{EQK1997}, Grigorova et al.~\cite{GIOOQ2017} and Quenez and Sulem~\cite{QS2014}. Although it is common to set $S=W$ and $A=0$, an extension
to a more general situation is also feasible when the generator $g$ satisfies a suitable Lipschitz-type condition (see \cite{NR2} for the case of a BSDE without reflection).}
\end{remark}

The first main result in this section is the following theorem.

\begin{theorem} \label{the4.1}
 Let  Assumptions \ref{assnBSDE}--\ref{assBSDEm2} be satisfied and let $(Y,Z,\Kpp)$ be the unique solution to the issuer's reflected BSDE \eqref{iRBSDE}. Then the following assertions are valid:\\
(i) if $\cEgi$ has the strict comparison property, then
\begin{equation}  \label{efq4.3}
\breve{p}^{r,i}(x_1,\cC^a)=\breve{p}^{s,i}(x_1,\cC^a)=Y_0-x_1=\cEgi_{0,\tau^i }(X_{\tau^i})-x_1=\wh{v}^i(\cC^a)-x_1
\end{equation}
and $(p',\phi',\tau')=(Y_0-x_1,Z,\tau^i)$ is an issuer's replicating strategy for $\cC^a$, \\
(ii) if $\cEgi$ has the strict comparison property, then the issuer's acceptable price for $\cC^a$ equals $p^i(x_1,\cC^a)=Y_0-x_1=\wh{v}^i(\cC^a)-x_1$.
\end{theorem}

\begin{proof}
The proof of Theorem \ref{the4.1} is split into three steps, which are formulated as Propositions \ref{pro4.1}, \ref{pro4.2} and \ref{pro4.3}.
\end{proof}

We start by analyzing the issuer's minimum superhedging cost (note that Assumption \ref{assBSDEm2} is not postulated in Proposition \ref{pro4.1}).

\begin{proposition} \label{pro4.1}
If Assumptions  \ref{assnBSDE}--\ref{assBSDEm1} are satisfied and $\cEgi$ has the comparison property, then the issuer's
minimum superhedging cost is well defined and satisfies
\begin{equation} \label{eq4.2}
\breve{p}^{s,i}(x_1,\cC^a)=\vsup^i(\cC^a)-x_1=Y_0-x_1,
\end{equation}
where $(Y,Z,\Kpp)$ is the unique solution to the reflected BSDE \eqref{iRBSDE}.
\end{proposition}

\proof
We first prove that $\pinf^{s,i}\leq Y_0-x_1$. It suffices to show that for the initial value $p':=Y_0-x_1$, where $Y_0$ is obtained from the reflected BSDE \eqref{iRBSDE}, we can find an issuer's superhedging strategy, that is, there exists a trading strategy $\phi'\in\Psi(x_1+p',A)$ such that $V_t(x_1+p',\phi')\geq X_t$ for all $t\in[0,T]$. To this end, we set $(p',\phi')=(Y_0-x_1,Z)$ where $(Y,Z,\Kpp)$ is the unique solution to the reflected BSDE \eqref{iRBSDE}. Then, on the one hand, the value process $V=V(x_1+p',\phi')$ is a unique strong solution (by Assumption \ref{assnBSDE}(ii)) to the following SDE where the initial value $V_0=Y_0$ and the process $Z$ are fixed
\begin{equation}\label{value process SDE}
dV_t=-g(t,V_t,Z_t)\,dt+Z_t\,dS_t+dA_t.
\end{equation}
On the other hand, if $(Y,Z,\Kpp)$ solves the reflected BSDE \eqref{iRBSDE}, then the process $\wt{Y}=Y$ can also be seen as a unique strong solution to the following SDE
\begin{equation} \label{forward issuer RBSDE}
d\wt{Y}_t=-g(t,\wt{Y}_t,Z_t)\,dt+Z_t\,dS_t+dA_t-d\Kpp_t,
\end{equation}
where, once again, the initial value $\wt{Y}_0=Y_0$ and the processes $Z$ and $\Kpp$ are given. Therefore, from Lemma \ref{lem:1} with $g_1=g_2=g$ we infer that $V_t\ge \wt{Y}_t=Y_t$ for all $t\in[0,T]$. Since $Y_t \ge X_t$ for all $t\in [0,T]$, we conclude that $V_t\ge X_t$ for all $t\in[0,T]$. Consequently, $(x_1+p',\phi')=(Y_0,Z)$ is an issuer's superhedging strategy and thus $\pinf^{s,i}\leq Y_0-x_1$.

We will now show that $\pinf^{s,i}\ge Y_0-x_1$. Let us consider an arbitrary $p\in\rr$ for which there exists $\phi\in\Psi(x_1+p,A)$ such that $(p,\phi)$ satisfy (SH). If we can show that $x_1+p\ge Y_0$, then the inequality $\pinf^{s,i} \ge Y_0-x_1$ will hold by the definition of the lower bound $\pinf^{s,i}$. To this end, we observe that $V_{\tau}(x_1+p,\phi)\geq X_{\tau}$ for every $\tau\in\cT $ since, by Definition \ref{def1.2}, we have that $V_t(x_1+p,\phi)\geq X_t$ for all $t \in [0,T]$. Consequently, by applying the mapping $\cEgi$ to both sides and using the comparison property of $\cEgi$, we obtain
\[
x_1+p=\cEgi_{0,\tau}(V_{\tau}(x_1+p,\phi))\geq \cEgi_{0,\tau}(X_{\tau}).
\]
Since $\tau\in\cT $ is arbitrary, we conclude that $x_1+p\geq \sup_{\tau\in\cT}\cEgi_{0,\tau}(X_{\tau})=\vsup^i(\cC^a)=Y_0$ where the second equality follows from Assumption \ref{assBSDEm1}. Hence $\pinf^{s,i}\geq Y_0-x_1$ and thus we conclude that the equality $\pinf^{s,i}=Y_0-x_1$ is valid.

Finally, from the first part of the proof, we know that for $p'=Y_0-x_1$ there exists a trading strategy $\phi'=Z\in\Psi(x_1+p',A)$ such that $V_t(x_1+p',\phi')\geq X_t$ for all $t\in[0,T]$ so that $Y_0-x_1\in \cH^{s,i}(x_1)$. Consequently, we have that $\pinf^{s,i}(x_1,\cC^a)=\breve{p}^{s,i}(x_1,\cC^a)=Y_0-x_1$.
\endproof

Consider the solution $(Y,Z,\Kpp)$ and the pair $(Y_0,Z)$ introduced in the proof of Proposition \ref{pro4.1}. In the next result, we examine the existence of an issuer's replicating strategy for $\cC^a$.

\begin{proposition} \label{pro4.2}
If Assumptions  \ref{assnBSDE}--\ref{assBSDEm2} are satisfied and $\cEgi$ has the strict comparison property, then the following assertions are valid: \\
(i) the pair $(Y_0-x_1,Z)$ is an issuer's replicating strategy for $\cC^a$ and $\tau^i$ is an issuer's break-even time for the pair $(Y_0-x_1,Z)$, \hfill \break
(ii) the issuer's minimum superhedging and replication costs satisfy
\begin{equation}  \label{eq4.3}
\breve{p}^{r,i}(x_1,\cC^a)=\breve{p}^{s,i}(x_1,\cC^a)=Y_0-x_1=\cEgi_{0,\tau^i }(X_{\tau^i})-x_1=\wh{v}^i(\cC^a)-x_1.
\end{equation}
\end{proposition}

\proof
From Lemma \ref{lem1.7} and Proposition \ref{pro4.1}, we already know that $Y_0-x_1=\pinf^{s,i}=\check{p}^{s,i} \leq\pinf^{r,i}$ and thus to establish the equality $\breve{p}^{r,i}(x_1,\cC^a)=\breve{p}^{s,i}(x_1,\cC^a)$, it is enough to show that the trading strategy $(p',\phi')=(Y_0-x_1,Z)$, which is already known to be an issuer's superhedging strategy, is also an issuer's replicating strategy for $\cC^a$. We first note that the definition of $\tau^i$ and the right-continuity of the processes $X$ and $Y$ yield the equality $X_{\tau^i}=Y_{\tau^i}$. Consequently, we have that
\[
Y_0=\wh{v}^i(\cC^a)=\cEgi_{0,\tau^i }(X_{\tau^i})= \cEgi_{0,\tau^i }(Y_{\tau^i}),
\]
where the first two equalities follow from Assumptions \ref{assBSDEm1} and \ref{assBSDEm2}. We will now show that $\Kpp_{\tau^i}=0$.
Since $(Y,Z,\Kpp)$ solves the reflected BSDE \eqref{iRBSDE},
we also know that
\begin{equation*}
Y_0=Y_{\tau^i} +\int_0^{\tau^i} g(u,Y_u,Z_u)\,du-\int_0^{\tau^i} Z_u\,dS_u-A_{\tau^i}+K_{\tau^i}.
\end{equation*}
Therefore, $Y_0=\cEgi_{0,\tau^i }(Y_{\tau^i}+K_{\tau^i})$ and thus $\cEgi_{0,\tau^i }(Y_{\tau^i})=\cEgi_{0,\tau^i }(Y_{\tau^i}+K_{\tau^i})$.
From the strict comparison property of $\cEgi$, we conclude that $K_{\tau^i}=0$ and thus, for all $t \in [0,\tau^i]$,
\begin{equation*}
Y_t=Y_0-\int_0^t g(u,Y_u,Z_u)\,du+\int_0^t Z_u\,dS_u + A_t.
\end{equation*}
Finally, using the equality $V_0(Y_0,Z)=Y_0$ and the postulated uniqueness of a solution to the SDE (\ref{value process SDE}) (see Assumption \ref{assnBSDE}(ii)), we obtain the equality  $V_t(Y_0,Z)=Y_t$ on $[0,\tau^i]$ and thus, in particular, $V_{\tau^i}(Y_0,Z)=Y_{\tau^i}=X_{\tau^i}$. We have thus shown that $\tau^i$ is an issuer's break-even time for the pair $(Y_0-x_1,Z)$.  We have thus shown that the pair $(p',\phi' )=(Y_0,Z)$ is an issuer's replicating strategy for $\cC^a$.
\endproof

The final step in establishing Theorem \ref{the4.1} and hence providing a solution to the issuer's valuation problem hinges on showing that the issuer's acceptable price is also the issuer's maximum fair price.

\begin{proposition} \label{pro4.3}
If Assumptions  \ref{assnBSDE}--\ref{assBSDEm2} are satisfied and
$\cEgi$ has the strict comparison property, then the issuer's acceptable price $p^i(x_1,\cC^a)$ is well defined and
\begin{equation}  \label{eq4.4}
p^i(x_1,\cC^a)=\wh{p}^{f,i}(x_1,\cC^a)=\breve{p}^{r,i}(x_1,\cC^a)=\breve{p}^{s,i}(x_1,\cC^a).
\end{equation}
\end{proposition}

\proof
It suffices to show that $\breve{p}^{r,i}$ belongs to $\cH^{f,i}(x_1)$ or, equivalently, that $\breve{p}^{r,i}<p$
for every $p\in\cH^{a,i}(x_1)$ (recall that $\cH^{a,i}(x_1)$ is the complement of $\cH^{f,i}(x_1)$). To this end, we will argue by contradiction. Assume that $\breve{p}^{r,i}\in\cH^{a,i}(x_1)$ so that there exists a strategy $\breve{\phi}\in\Psi (x_1+\breve{p}^{r,i})$
such that $(\breve{p}^{r,i},\breve{\phi})$ satisfy (AO). Then we have, for every $\tau\in\cT$,
\[
\P\big(V_{\tau}(x_1+\breve{p}^{r,i},\breve{\phi})\geq X_{\tau}\big)=1\ \ \text{\rm and }\ \P\big(V_{\tau}(x_1+\breve{p}^{r,i},\breve{\phi})> X_{\tau}\big)>0.
\]
Let us now take $\tau=\tau^i$. By applying the mapping $\cEgi$ to both sides, we obtain
\[
x_1+\breve{p}^{r,i}=\cEgi_{0,\tau^i}\big(V_{\tau^i}(x_1+\breve{p}^{r,i},\breve{\phi})\big)>\cEgi_{0,\tau^i}(X_{\tau^i})=x_1+\breve{p}^{r,i},
\]
where the last equality comes from Proposition \ref{pro4.2}. This is a contradiction and thus we have shown that $\breve{p}^{r,i}$ is not in $\cH^{a,i}(x_1)$. Recall that either $\cH^{a,i}(x_1)=[\pinf^{a,i},\infty )$ or $\cH^{a,i}(x_1)=(\pinf^{a,i},\infty)$ and we claim
that in fact the latter case is true. Indeed, from Assumption \ref{assnBSDE}, Lemma \ref{lem1.4} and Proposition \ref{pro4.2}, we have $\breve{p}^{r,i}=\breve{p}^{s,i}=\pinf^{a,i}$ and, since $\breve{p}^{r,i}$ is not in $\cH^{a,i}(x_1)$, we have that $\cH^{a,i}(x_1)=(\pinf^{a,i},\infty)$. Obviously, $\breve{p}^{r,i}<p$ for every $p\in\cH^{a,i}(x_1)$ and thus $\breve{p}^{r,i}$ belongs to $\cH^{f,i}(x_1)$ meaning that the set $\cH^{f,r,i}(x_1)$ is nonempty. All equalities in \eqref{eq4.4} now follow from Proposition \ref{pro1.1}.
\endproof

\subsection{Issuer's Break-Even Times} \label{sec4.3b}

In Section \ref{sec4.3b}, we postulate that the assumptions of Theorem \ref{the4.1}(ii) are satisfied and the contract $\cC^a$ is traded at time 0 at the issuer's acceptable price $p^i=p^i(x_1,\cC^a)$. From Definition \ref{def1.15} and Propositions \ref{pro4.2} and \ref{pro4.3}, we know that there exists a pair $(\phi',\tau^i)\in\Psi(x_1+p^i,A)\times\cT$ such that $(p^i,\phi')$ satisfy (SH) and $(p^i,\phi',\tau^i)$ satisfy (BE), specifically, $\phi'=Z$ and $p^i=Y_0-x_1$, where $(Y,Z,\Kpp)$ is the unique solution to the reflected BSDE \eqref{iRBSDE}. Our next goal is to provide a detailed characterization of all issuer's break-even times associated with an issuer's replicating strategy $(p^i,\phi')$. To establish Theorem \ref{the4.2}, which is the second main result in Section \ref{sec4.3}, we will need the following additional assumption.

\begin{assumption} \label{assBSDEcomparison}
{\rm The following {\it extended comparison property} for solutions to BSDEs holds: if for $j=1,2$
\[
\left\{ \begin{array} [c]{ll}
dY_s^j=-g_j(s,Y^j_s,Z^j_s)\,ds+Z_s^j\,dS_s-d\hHh_s^j,\medskip\\ Y^j_{\tau}=\xi_j, \end{array} \right.
\]
where $\tau\in\cT$, $\xi_1\ge\xi_2,\,g_1(s,Y^2_s,Z^2_s)\geq g_2(s,Y^2_s,Z^2_s)$ for all $s\in[0,\tau ]$ and the process $\hHh^1-\hHh^2$ is nondecreasing, then $Y_s^1\geq Y_s^2$ for every $s\in[0,\tau ]$.}
\end{assumption}

\begin{remark}
{\rm Suitable versions of the comparison theorem for BSDEs are known and thus Assumption \ref{assBSDEcomparison} can be
checked to be met in several nonlinear market models 
(see, for instance, explicit examples analyzed in Nie and Rutkowski \cite{NR2}).}
\end{remark}

Before stating the main result in Section \ref{sec4.3},  let us recall the following well-known definition related to nonlinear
evaluations (see, e.g., Peng \cite{P2004a,P2004b}).

\begin{definition}
{\rm We say that a $\gg$-adapted, c\`adl\`ag 
process $Y$ is an $\cEgi$-{\it supermartingale} (respectively, an $\cEgi$-{\it martingale})
on $[0,T]$ if $Y_{s}\geq\cEgi_{s,t}(Y_t)$ (respectively, $Y_{s}=\cEgi_{s,t}(Y_t)$) for $0\leq s \leq t \leq T$.}
\end{definition}

\begin{theorem}  \label{the4.2}
Let Assumptions \ref{assnBSDE}--\ref{assBSDEcomparison} be satisfied and the strict comparison property of $\cEgi$ hold.
Then for the process $\phi'=Z \in \Psi(x_1+p^i,A)$ and an arbitrary $\tau'\in\cT $ the following assertions are equivalent:\\
(i) $\,\tau'$ is an issuer's break-even time for the pair $(p^i,\phi')\in\rr\times\Psi(x_1+p^i,A)$,  \\
(ii) the triplet $(p^i,\phi',\tau')$ satisfies {\rm (NA)},  \\
(iii) the equality $V_{\tau'}(x_1+p^i,\phi')=X_{\tau'}$ holds,\\
(iv) $\, X_{\tau'}=Y_{\tau'}$ and $\Kpp_{\tau'}=0$ and thus $Y$ is an $\cEgi$-martingale on $[0,\tau']$,\\
(v) $\,\tau'$ is a solution to the issuer's optimal stopping problem, that is, $\cEgi_{0,\tau' }(X_{\tau'})=\wh{v}^i(\cC^a)$. \\
The stopping time $\tau^i=\inf\,\{t\in [0,T]\,|\,Y_t=X_t\}$ is the earliest issuer's break-even time for $(p^i,\phi' )$.
\end{theorem}

\proof
Recall that if $\phi'=Z$, then the pair $(p^i,\phi')$ is an issuer's superhedging strategy for $\cC^a$ (see the proof of Proposition \ref{pro4.1}).
It is thus is clear that assertions (i), (ii) and (iii) are equivalent.

\noindent  (iii)$\, \Rightarrow\, $(iv). From the proof of Proposition \ref{pro4.1}, we already know that $V_t(x_1+ p^i,\phi')\geq Y_t\geq X_t$ for all $t\in[0,T]$ and thus, in particular, the inequality $V_{\tau}(x_1+p^i,\phi')\ge Y_{\tau}\ge X_{\tau}$ holds for every $\tau\in\cT$.  Since we assumed that (iii) holds, we have $V_{\tau'}(x_1+p^i,\phi')=X_{\tau'}$ and thus $V_{\tau'}(Y_0,\phi')=Y_{\tau'}=X_{\tau'}$ (recall from Theorem \ref{the4.1}(ii) that $p^i=Y_0-x_1$). It thus remains to show that $\Kpp_{\tau'}=0$. Since the process $V=V(Y_0,\phi')$ satisfies the SDE (\ref{value process SDE}), it is an $\cEgi$-martingale and thus we obtain the following equalities
\begin{equation}\label{equation 1}
\cEgi_{0,\tau'}(Y_{\tau'})=\cEgi_{0,\tau'}\big(V_{\tau'}(Y_0,\phi')\big)=Y_0.
\end{equation}
On the one hand, for any fixed $t \in (0,T]$ the process $\bar{Y}_s:=\cEgi_{s,t}(Y_{t}),\,s\in [0,t]$ solves the following BSDE
\[
\left \{\begin{array} [c]{ll}
d\bar{Y}_s=-g(s,\bar{Y}_s,\bar{Z}_s)\,ds+\bar{Z}_s\,dS_s+dA_s,\medskip\\ \bar{Y}_t=Y_t.
\end{array} \right.
\]
On the other hand, if $(Y,Z,\Kpp)$ solves the reflected BSDE \eqref{iRBSDE}, then for any fixed $[0,t]$, the pair $(\wt{Y},\wt{Z})=(Y,Z)$ is a unique solution to the BSDE
\[
\left \{\begin{array} [c]{ll}
d\wt{Y}_s=-g(s,\wt{Y}_s,\wt{Z}_s)\,ds+\wt{Z}_s\,dS_s+dA_s-d\Kpp_s,\medskip\\ \wt{Y}_t=Y_t ,
\end{array} \right.
\]
where $\Kpp$ is a predetermined increasing process. Therefore, in view of the extended comparison property of solutions to BSDEs (see Assumption \ref{assBSDEcomparison}), the inequality $\cEgi_{s,t}(Y_{t})\leq Y_t$ holds for all $s \in [0,t]$ and thus $Y$ is an $\cEgi$-supermartingale. Moreover, similarly to the above discussion, by using the extended comparison property of solutions to BSDEs, one can show that, for any $\theta\in \cT$, the inequality $\cEgi_{s,\theta}(Y_{\theta})\leq Y_s$ holds for all $s \in [0,\theta]$.
We will now show that for every $0\leq s\leq \tau'$
\begin{equation}\label{equationmore1}
\cEgi_{s,\tau'}(Y_{\tau'})=Y_s.
\end{equation}
Let us assume, on the contrary, that equality \eqref{equationmore1} fails to hold. Then, using the strict comparison property of $\cEgi$, we obtain $\cEgi_{0,\tau'}(Y_{\tau'})=\cEgi_{0,s}(\cEgi_{s,\tau'}(Y_{\tau'}))<\cEgi_{0,s}(Y_s)\leq Y_0$. This manifestly contradicts \eqref{equation 1} and thus \eqref{equationmore1} is valid. For every $0\leq s\leq t\leq \tau'$, from \eqref{equationmore1}, we have $\cEgi_{t,\tau'}(Y_{\tau'})=Y_t$ and then
\[
\cEgi_{s,t}(Y_{t})=\cEgi_{s,t}(\cEgi_{t,\tau'}(Y_{\tau'}))=\cEgi_{s,\tau'}(Y_{\tau'})=Y_s,
\]
where the last equality also comes from  \eqref{equationmore1}. We conclude that $Y$ is an $\cEgi$-martingale on $[0,\tau' ]$ and thus the equality $\Kpp_{\tau'}=0$ follows.

\noindent (iv)$\, \Rightarrow \,$(iii). By assumption, $Y_{\tau'}=X_{\tau'}$ and $\Kpp_{\tau'}=0$ and thus the reflected BSDE \eqref{iRBSDE} reduces to the following BSDE on $[0,\tau']$
\[
\left\{ \begin{array} [c]{ll}
dY_t=-g(t,Y_t,Z_t)\,dt+Z_t\,dS_t+dA_t,\medskip\\ Y_{\tau'}=X_{\tau'}.
\end{array} \right.
\]
The above BSDE can also be represented in the forward manner, for $t \in[0,\tau']$,
\[
dY_t=-g(t,Y_t,Z_t)\,dt+Z_t\,dS_t+dA_t,
\]
where the initial value $Y_0$ and the process $Z$ are given. Similarly, the wealth process $V:=V(x_1+p^i,\phi')=V(Y_0,Z)$ solves the following SDE, for $t\in[0,T]$,
\[
dV_t=-g(t,V_t,Z_t)\,dt+Z_t\,dS_t+dA_t,
\]
with initial condition $V_0=Y_0$. From the uniqueness of a solution to the above SDE, we infer that $V_t=Y_t$ for $t\in[0,\tau']$. In particular, $V_{\tau'}(x_1+p^i,\phi' )=Y_{\tau'}=X_{\tau'}$, as was required to show.

\noindent (iv)$\, \Rightarrow \,$(v).  From the $\cEgi$-martingale property of $Y$ on $[0,\tau']$, we get $\cEgi_{0,\tau'}(Y_{\tau'})=Y_0$. In view of Assumption \ref{assBSDEm1}, we have that $Y_0=\vsup^i(\cC^a)$ and thus the equalities $\cEgi_{0,\tau'}(X_{\tau'})=\vsup^i(\cC^a)=\wh{v}^i(\cC^a)$ hold, which means that $\tau'$ is a solution to the issuer's optimal stopping problem.

\noindent (v)$\, \Rightarrow \,$(iv). From condition (v) and Assumption \ref{assBSDEm1}, we obtain the equality $Y_0=\cEgi_{0,\tau'}(X_{\tau'})$. We will now use similar arguments as in the proof of the implication (iii)$\, \Rightarrow\, $(iv). First, the process $\bar{X}_s:=\cEgi_{s,t}(X_{\tau'})$ solves the following BSDE
\[
\left \{\begin{array} [c]{ll}
d\bar{X}_s=-g(s,\bar{X}_s,\bar{Z}_s)\,ds+\bar{Z}_s\,dS_s+dA_s,\medskip\\ \bar{X}_{\tau'}=X_{\tau'},
\end{array} \right.
\]
and it also satisfies $\bar{X}_0=Y_0$.  Second, if $(Y,Z,\Kpp)$ solves the reflected BSDE \eqref{iRBSDE}, then the pair $(\wt{Y},\wt{Z})=(Y,Z)$ is a unique solution to the BSDE
\[
\left \{\begin{array} [c]{ll}
d\wt{Y}_s=-g(s,\wt{Y}_s,\wt{Z}_s)\,ds+\wt{Z}_s\,dS_s+dA_s-d\Kpp_s,\medskip\\ \wt{Y}_{\tau'}=Y_{\tau'} \geq X_{\tau'} ,
\end{array} \right.
\]
where $\Kpp$ is a predetermined increasing process and, obviously, $\wt{Y}_0=Y_0$.  The extended comparison property of solutions to BSDEs yields the inequality $\cEgi_{t,\tau'}(Y_{\tau'})\leq Y_t$ for all $t \in [0,\tau']$. Therefore, if $Y_{\tau'} \geq X_{\tau'}$ and $Y_{\tau'} \ne X_{\tau'}$, then the strict comparison property of $\cEgi$ gives
\[
Y_0 \geq \cEgi_{0,\tau'}(Y_{\tau'}) > \cEgi_{0,\tau'}(X_{\tau'})=Y_0,
\]
which is a contradiction. This shows that $Y_{\tau'}=X_{\tau'}$. As in the proof of the implication (iii)$\, \Rightarrow\, $(iv), we argue that $Y$ is an $\cEgi$-martingale on $[0,\tau' ]$ and thus the equality $\Kpp_{\tau'}=0$ is satisfied.

It remains to show that the last assertion is valid. In view of Assumption \ref{assBSDEm2}, $\tau^i$ is a solution to the issuer's optimal stopping problem and thus, from part (v),  $\tau^i$ is an issuer's break-even time for $(p^i,\phi')$. In view of (iv), for any break-even time for for $(p^i,\phi')$ we have that $X_{\tau'}=Y_{\tau'}$. The definition of $\tau^i$ now shows that it is the earliest issuer's break-even time for $\cC^a$.
\endproof


\subsection{Holder's Acceptable Price via RBSDE}   \label{sec4.4}

We now address the pricing, hedging and exercising problems from the perspective of the holder. The corresponding nonlinear evaluation $\cE^{g,-A}$, which is defined through solutions to the holder's BSDE
\begin{equation} \label{ncBSDE}
y_t=\zeta_s+\int_t^s g(u,y_u,z_u)\,du-\int_t^s z_u\,dS_u+A_s-A_t,
\end{equation}
is henceforth denoted by $\cEgh$ and called the {\it holder's $g$-evaluation}. For brevity, we denote  by $x:=\Vben(x_2)+\Xhh$
the {\it holder's relative reward} and we assume that the process $x$ is square-integrable.

\begin{assumption} \label{assBSDEh}
{\rm The reflected BSDE with the upper obstacle $x$
\begin{equation} \label{hRBSDE}
\left\{ \begin{array}[c]{ll}
dy_t=-g(t,y_t,z_t)\,dt+z_t\,dS_t-dA_t+d\kpp_t,\medskip\\
y_T=x_T,\quad y_t\leq x_t,\quad \int_0^T (x_t-y_t)\,d\kpp^c_t=0,\quad \Delta \kpp^{d}_t = \Delta (y_t+A_t)\I_{\{ y_{t-}= x_{t-}\}},
\end{array} \right.
\end{equation}
has a unique solution $(y,z,\kpp)$ where $\kpp$ is a $\gg$-predictable, nondecreasing process such that $\kpp_0=0$ and $\kpp=\kpp^c+\kpp^d$ is its unique decomposition into continuous and jump components.}
\end{assumption}

\begin{definition} \label{cef2vxa} {\rm
We say that $\vinf^h(\cC^a)\in\rr$ is the {\it value} of the holder's optimal stopping problem for $\cC^a$ if
\begin{equation*}
\vinf^h(\cC^a)=\inf_{\tau\in\cT}\,\cEgh_{0,\tau }(x_{\tau}).
\end{equation*}
A stopping time $\taus \in\cT$ is called a {\it solution} to the holder's optimal stopping problem if $\vinf^h(\cC^a)=\breve{v}^h(\cC^a)$ where}
\begin{equation}  \label{ceq2vb}
\breve{v}^h(\cC^a):=\cEgh_{0,\taus }(x_{\taus})=\min_{\tau\in\cT}\,\cEgh_{0,\tau }(x_{\tau}).
\end{equation}
\end{definition}

The following assumptions can be justified by an independent analysis of a nonlinear optimal stopping problem.
Although we take here the properties stated in Assumptions \ref{assBSDEv1} and \ref{assBSDEv2} for granted, it is worth to mention
that they are supported by recent results in Grigorova et al.~\cite{GIOOQ2017,GIOQ2017}.

\begin{assumption} \label{assBSDEv1}
{\rm The value $\vinf^h(\cC^a)$ to the holder's optimal stopping problem exists and satisfies $\vinf^h(\cC^a)=y_0$.}
\end{assumption}

\begin{assumption} \label{assBSDEv2}
{\rm The stopping time $\tau^h :=\inf\,\{t\in [0,T]\,|\,y_t=x_t\}$ is a solution to the holder's optimal stopping problem.}
\end{assumption}

A salient feature of an American contract is a holder's {\it rational exercise time}, which in our framework is defined as follows.

\begin{definition} \label{holderrationaltime}
{\rm We say that $\tau\in\cT$ is a \textit{rational exercise time} for the holder of $\cC^a$ if the contract is traded at the holder's maximum superhedging cost $\wh{p}^{s,h}=\wh{p}^{s,h}(x_2,\cC^a)$ at time 0 and there exists a strategy $\psi\in\Psi(x_2-\wh{p}^{s,h},-A)$ such that $V_{\tau}(x_2-\wh{p}^{s,h},\psi)\geq x_{\tau}$.}
\end{definition}

In fact, we will use Definition \ref{holderrationaltime} within the setup where the equality $\wh{p}^{r,h}(x_2,\cC^a)=\wh{p}^{s,h}(x_2,\cC^a)$ holds. If, in addition, the strict comparison property for the BSDE with the driver $g$ is satisfied, then the inequality $V_{\tau}(x_2-\wh{p}^{r,h},\psi)\geq x_{\tau}$ can be replaced by the equality $V_{\tau}(x_2-\wh{p}^{r,h},\psi)=x_{\tau}$ so that a rational exercise time is also a holder's break-even time.
The following theorem describes the properties of a solution to the holder's pricing problem for an American contract $\cC^a$.

\begin{theorem} \label{the4.3}
Let Assumption \ref{assnBSDE} and \ref{assBSDEh}--\ref{assBSDEv2} be satisfied and let $(y,z,\kpp)$ be a unique solution to the holder's reflected BSDE
\eqref{hRBSDE}. Then the following assertions are valid.\\
(i) If $\cEgh$ has the strict comparison property, then
\begin{equation} \label{uuii}
\wh{p}^{r,h}(x_2,\cC^a)= \wh{p}^{s,h}(x_2,\cC^a)=x_2-y_0=x_2-\cEgh_{0,\tau^h}(x_{\tau^h})=x_2-\breve{v}^h(\cC^a).
\end{equation}
A holder's replicating strategy for $\cC^a$ is given by the triplet $(p',\psi',\tau')=(x_2-y_0,z,\tau^h)$ and $\tau^h$ is a holder's rational exercise time.\\
(ii) If $\cEgh$ has the strict comparison property, then the holder's acceptable price for $\cC^a$ equals
$p^h(x_2,\cC^a)=x_2-y_0=x_2-\breve{v}^h(\cC^a)$.
\end{theorem}

\begin{proof}
As in the case of the issuer, we split the proof into three steps, which are formulated as Propositions \ref{pro4.4}, \ref{pro4.5}
and \ref{pro4.6}. They can be seen as the holder's counterparts of Propositions \ref{pro4.1}, \ref{pro4.2} and \ref{pro4.3}, although their statements and proofs differ from the issuer's case.
\end{proof}

Note the Assumption \ref{assBSDEv2} is not postulated in the next result. It is worth stressing that the equality $\psup^{s,h}(x_2,\cC^a)=x_2-y_0$ does not necessarily hold under the assumptions of Proposition \ref{pro4.4} and the holder's maximum superhedging cost $\wh{p}^{s,h}(x_2,\cC^a)$ is not necessarily well defined.


\begin{proposition} \label{pro4.4}
If Assumptions \ref{assnBSDE} and \ref{assBSDEh}--\ref{assBSDEv1} are satisfied and $\cEgh$ has the comparison property, then
\begin{equation} \label{4.5ce}
\psup^{s,h}(x_2,\cC^a) \leq x_2-\vinf^h(\cC^a)=x_2-y_0,
\end{equation}
where $(y,z,\kpp)$ is the unique solution to the reflected BSDE \eqref{hRBSDE}.
\end{proposition}

\proof
We will show that $\psup^{s,h} \leq x_2-y_0$. By the definition of the supremum, it is enough to show that $x_2-y_0 \geq p$ for all $p\in\cH^{s,h}(x_2)$. From the definition of $\cH^{s,h}(x_2)$, we know that for any $p\in\cH^{s,h}(x_2)$, there exists a pair $(\psi,\tau)\in\Psi(x_2-p,-A)\times\cT$ such that $V_{\tau}(x_2-p,\psi)\ge x_{\tau}$. The comparison property of $\cEgh$ gives
\[
x_2-p=\cEgh_{0,\tau}\big(V_{\tau}(x_2-p,\psi)\big)\geq\cEgh_{0,\tau}(x_{\tau})
\]
and thus
\[
x_2-p\geq\inf_{\tau\in\cT}\cEgh_{0,\tau}\big(V_{\tau}(x_2-p,\psi)\big)\geq\inf_{\tau\in\cT}\cEgh_{0,\tau}(x_{\tau})=\vinf^h(\cC^a)=y_0,
\]
where the last equality follows from Assumption \ref{assBSDEv1}. We have thus shown that $\psup^{s,h}(x_2,\cC^a)\leq x_2-y_0=x_2-\vinf^h(\cC^a)$.
\endproof

We will now give conditions under which, in particular, the holder's maximum superhedging and replication costs are well defined and we show that they are equal, that is,  $\wh{p}^{r,h}(x_2,\cC^a)=\wh{p}^{s,h}(x_2,\cC^a)$.

\begin{proposition} \label{pro4.5}
If Assumptions \ref{assnBSDE} and \ref{assBSDEh}--\ref{assBSDEv2} are satisfied and $\cEgh$ has the strict comparison property, then: \\
\noindent (i) $(x_2-y_0,z,\tau^h)$ is a holder's replicating strategy for $\cC^a$,\\
\noindent (ii) the holder's maximum replication cost is well defined and satisfies
\[
\wh{p}^{r,h}(x_2,\cC^a)=\wh{p}^{s,h}(x_2,\cC^a)=x_2-y_0=x_2-\breve{v}^h(\cC^a)=x_2-\cEgh_{0,\tau^h}(x_{\tau^h}).
\]
\end{proposition}

\proof
We already know that $x_2-y_0 \geq \psup^{s,h}\geq \psup^{r,h}$  (see \eqref{cequac} and \eqref{4.5ce}). Hence to prove parts (i) and (ii), it suffices to show that if $(y,z,\kpp)$ is the unique solution to the reflected BSDE \eqref{hRBSDE}, then $(p',\psi',\tau')=(x_2-y_0,z,\tau^h)$ is a holder's replicating strategy. The wealth process $V=V(x_2-p',\psi')$ satisfies the SDE
\begin{equation} \label{eq.fSDE}
dV_t=-g(t,V_t,z_t)\,dt+z_t\,dS_t-dA_t,
\end{equation}
where the initial value $V_0=y_0$ and the process $z$ are given. The definition of $\tau^h$ and the right-continuity of the processes $x$ and $y$ ensure that $x_{\tau^h}=y_{\tau^h}$ so that
\[
y_0=\breve{v}^h(\cC^a)=\cEgh_{0,\tau^h }(x_{\tau^h})= \cEgh_{0,\tau^h }(y_{\tau^h}),
\]
where the second equality is a consequence of Assumption \ref{assBSDEv2}, and thus we see that $y_0=\cEgh_{0,\tau^h }(y_{\tau^h})$.
Therefore, using the strict comparison property of $\cEgh$ and simple arguments analogous to those used in the derivation of the equality $\Kpp_{\tau^i}=0$ in the proof of Proposition \ref{pro4.2}, we obtain the equality $\kpp_{\tau^h}=0$. Since $\kpp_t=0$ on $[0,\tau^h]$, the reflected BSDE \eqref{hRBSDE} can be seen on $[0,\tau^h]$ as the forward SDE
\begin{equation} \label{eq.fbSDE}
dy_t=-g(t,y_t,z_t)\,dt+z_t\,dS_t-dA_t,
\end{equation}
where the initial value $y_0=V_0$ and the process $z$ are given. From the uniqueness of a solution to the SDE \eqref{eq.fSDE}, it follows that $V=y$ on $[0,\tau^h]$. Hence $V_{\tau^h}=y_{\tau^h}=x_{\tau^h}$ and thus the triplet $(p',\psi',\tau')=(x_2-y_0,z,\tau^h)$ is indeed a holder's replicating strategy.
\endproof

To complete the proof of Theorem \ref{the4.4}, we need to examine the existence of the holder's acceptable price $p^h(x_2,\cC^a)$.
This will be done in the proof of the following proposition.

\begin{proposition} \label{pro4.6}
If Assumptions \ref{assnBSDE} and \ref{assBSDEh}--\ref{assBSDEv2} are satisfied and $\cEgh$ has the strict comparison property, then the holder's acceptable price $p^h(x_2,\cC^a)$ is well defined and
\[
p^h(x_2,\cC^a)=\breve{p}^{f,h}(x_2,\cC^a)=\wh{p}^{r,h}(x_2,\cC^a)=\wh{p}^{s,h}(x_2,\cC^a).
\]
\end{proposition}

\proof
We will first show that $\wh{p}^{r,h} \in\cH^{f,h}(x_2)$. In view of \eqref{firstn} and \eqref{secondn}, it is enough to prove that $\wh{p}^{r,h}(x_2,\cC^a)>p$ for every $p\in\cH^{a,h}(x_2)$. To this end, we argue by contradiction.  Let us write $\wh{p}=\wh{p}^{r,h}(x_2,\cC^a)$. Assume that $\wh{p}\in\cH^{a,h}(x_2)$ so that there exists $(\whphi,\wh{\tau})\in\Psi(x_2-\wh{p})\times\cT$ such that $(\wh{p},\whphi,\wh{\tau})$ satisfies (AO$'$), that is,
\[
\P\big(V_{\wh{\tau}}(x_2-\wh{p},\whphi)\geq x_{\wh{\tau}}\big)=1 \ \ \text{\rm and }\ \P\big(V_{\wh{\tau}}(x_2-\wh{p},\whphi)>x_{\wh{\tau}}\big)>0.
\]
By applying the mapping $\cEgh$, we obtain
\[
x_2-\wh{p}=\cEgh_{0,\wh{\tau}}\big(V_{\wh {\tau}}(x_2-\wh {p}^{r,h},\whphi)\big)>
\cEgh_{0,\wh{\tau}}(x_{\wh{\tau}})\geq \inf_{\tau\in\cT}\cEgh_{0,\tau}(x_{\tau})=\cEgh_{0,\tau^h}(x_{\tau^h})=x_2-\wh{p},
\]
where the last equality comes from Proposition \ref{pro4.5}. This is a contradiction and thus $\wh{p}^{r,h}(x_2,\cC^a) \notin \cH^{a,h}(x_2)$. In general, either $\cH^{a,h}(x_2)=(-\infty,\psup^{a,h}(x_2,\cC^a)]$ or $\cH^{a,h}(x_2)=(-\infty,\psup^{a,h}(x_2,\cC^a))$ and we argue that the latter situation occurs. Indeed, from Assumption \ref{assnBSDE}, Lemma \ref{lem1.5} and Propositions \ref{pro4.4} and \ref{pro4.5}, we have $\wh{p}^{r,h}(x_2,\cC^a)=\wh{p}^{s,h}(x_2,\cC^a)=\psup^{a,h}(x_2,\cC^a)$ and thus, since $\wh{p}^{r,h}(x_2,\cC^a)$ is not in $\cH^{a,h}(x_2)$, we conclude that $\cH^{a,h}(x_2)=(-\infty,\psup^{a,h}(x_2,\cC^a))$. It is also clear that $\wh{p}^{r,h}(x_2,\cC^a)>p$ for every $p\in\cH^{a,h}(x_2)$ and thus $\wh{p}^{r,h}(x_2,\cC^a)$ belongs to $\cH^{f,h}(x_2)$ so that $\cH^{f,r,h}(x_2)\neq\emptyset$. We complete the proof by making use of Proposition \ref{pro1.2}.
\endproof

\subsection{Holder's Rational Exercise Times}   \label{sec4.4b}

We conclude the paper by an analysis of the properties of holder's rational exercise times. Note that in Theorem  \ref{the4.4} we work under the assertions of Theorem \ref{the4.3}. It is thus known that the equality $\wh{p}^{r,h}(x_2,\cC^a)=\wh{p}^{s,h}(x_2,\cC^a)$ holds and thus a stopping time $\tau\in\cT$ is a holder's rational exercise time if the contract is traded at the holder's maximum replication cost $\wh{p}^{r,h}=\wh{p}^{r,h}(x_2,\cC^a)$ at time 0 and there exists a strategy $\psi\in\Psi(x_2-\wh{p}^{r,h},-A)$ such that $V_{\tau}(x_2-\wh{p}^{r,h},\psi)=x_{\tau}$. We thus deal here with a natural extension of the classical concept of a rational exercise time for the holder of an American option when the underlying market model is linear. Let us notice that in any complete linear market, albeit not in a general nonlinear market, any holder's rational exercise time is also a break-even time for the issuer (in particular, the equality $\tau^h=\tau^i$ is satisfied).

Recall that $\gg$-adapted, c\`adl\`ag process $Y$ is an $\cEgh$-submartingale (respectively, an $\cEgh$-martingale) on $[0,T]$
if $Y_{s } \leq\cEgh_{s,t}(Y_t)$ (respectively, $Y_{s}=\cEgh_{s,t}(Y_t)$) for all $0\leq s \leq t\leq T$.

The following result gives a characterisation of all holder's rational exercise times and describes the earliest and the
latest rational exercise times. Of course, results of this kind are well known from the existing literature on a classical
optimal stopping problem (see, for instance, Kobylanski and Quenez \cite{KQ2012}). For nonlinear optimal stopping problems, the interested reader is also referred to Dumitrescu \cite{DQS2017} and Grigorova et al. \cite{GIOOQ2017,GIOQ2017}.

\begin{theorem}  \label{the4.4}
Let Assumptions \ref{assnBSDE} and \ref{assBSDEcomparison}--\ref{assBSDEv2} be satisfied and
the strict comparison property of $\cEgh$ hold. In particular, let $(y,z,\kpp)$ be the unique solution to the reflected BSDE \eqref{hRBSDE}.
Then a stopping time $\tau' \in \cT$ is a holder's rational exercise time if and only if the following conditions are met: \\
(i) $y$ is a $\cEgh$-martingale on $[0,\tau']$, that is, $k_{\tau'}=0$,\\
(ii) the equality $y_{\tau'}=x_{\tau'}$ holds. \\
The earliest holder's rational exercise time equals $\tau^h:=\inf\,\{t\in [0,T]\,|\,y_t=x_t\}$. If, in addition, the process $\kpp$ is continuous, then $\bar{\tau}^h:=\inf\,\{t\in [0,T]\,|\,\kpp_t >0\}$ is the latest holder's rational exercise time.
\end{theorem}

\proof Let $\tau'\in\cT$ be any stopping time such that conditions (i) and (ii) are met. Since $y_{\tau'}=x_{\tau'}$ and $\kpp_{\tau'}=0$, we see that the triplet $(y,z,\kpp)$ solves the following BSDE on $[0,\tau']$
\[
\left\{ \begin{array} [c]{ll}
dy_t=-g(t,y_t,z_t)\,dt+z_t\,dS_t-dA_t,\medskip\\ y_{\tau'}=x_{\tau'},
\end{array} \right.
\]
which can also be written in the forward manner, for all $t\in[0,\tau']$,
\[
dy_t=-g(t,y_t,z_t)\,dt+z_t\,dS_t-dA_t,
\]
where initial condition $y_0$ and the process $z$ are given.
Now we take $\psi=z$ and we recall that $\wh{p}^{r,h}(x_2,\cC^a)=x_2-y_0$ (see Proposition \ref{pro4.6}). Hence the wealth process $V=V(x_2-\wh{p}^{r,h}(x_2,\cC^a),\psi)$ satisfies the following SDE for all $t\in[0,T]$
\[
dV_t=-g(t,V_t,z_t)\,dt+z_t\,dS_t-dA_t,
\]
with initial condition $V_0=y_0$. From the uniqueness of a solution to the above SDE, we infer that $V_t=y_t\leq x_t$ for every $t\in[0,\tau']$. In particular, $V_{\tau'}=y_{\tau'}=x_{\tau'}$ and thus $\tau'$ is a rational exercise time for the holder of $\cC^a$.

Let us now assume that $\tau'$ is a rational exercise time for the holder of $\cC^a$. From Definition \ref{holderrationaltime}, it follows that for $p=\wh{p}^{r,h}(x_2,\cC^a)=x_2-y_0$ there exists a strategy $\psi\in\Psi(x_2-p,-A)$ such that $V_{\tau'}(x_2-p,\psi )\geq x_{\tau'}$. The comparison property of $\cEgh$ yields
\begin{equation} \label{equation holder}
y_0=x_2-p=\cEgh_{0,\tau'}\big(V_{\tau'}(x_2-p,\psi)\big)\ge\cEgh_{0,\tau'}(x_{\tau'})\geq\cEgh_{0,\tau'}(y_{\tau'}),
\end{equation}
where the last inequality is valid since $x_{\tau'}\geq y_{\tau'}$. For any fixed $t\in (0,T]$, the process $\bar{y}_s:=\cEgh_{s,t}(y_t)$ solves the following BSDE on $[0,t]$
\[
\left\{ \begin{array} [c]{ll}
d\bar{y}_s=-g(s,\bar{y}_s,\bar{z}_s)\,ds+\bar{z}_s\,dS_s-dA_s,\medskip\\ \bar{y}_t=y_t,
\end{array} \right.
\]
and, if $(y,z,k)$ solves the reflected BSDE \eqref{hRBSDE}, then $y$ satisfies the following BSDE on $[0,t]$
\[
\left\{ \begin{array} [c]{ll}
dy_s=-g(s,y_s,z_s)\,ds+z_s\,dS_s-dA_s+d\kpp_s,\medskip\\ y_t=y_t. 
\end{array} \right.
\]
Using the extended comparison property postulated in Assumption \ref{assBSDEcomparison}, we get $y_s \leq \bar{y}_s = \cEgh_{s,t}(y_{t})$ for all $s\in[0,t]$ and thus $y$ is an $\cEgh$-submartingale. Furthermore,  using the extended comparison property of solutions to BSDEs, one can show that for any $\theta\in \cT$ we have $\cEgh_{s,\theta}(y_{\theta})\ge y_s$ for all $s \in [0,\theta]$. In view of \eqref{equation holder} and the strict comparison property of $\cEgh$, we deduce that for every $0\leq s\leq \tau'$
\begin{equation}\label{equationholdermore1}
\cEgh_{s,\tau'}(y_{\tau'})=y_s.
\end{equation}
Indeed, suppose that this is not true. Then the strict comparison property of $\cEgh$ would yield
\[
\cEgh_{0,\tau'}(y_{\tau'})= \cEgh_{0, s}(\cEgh_{s,\tau'}(y_{\tau'}))>\cEgh_{0, s}(y_s)\ge y_0,
 \]
which would contradict \eqref{equation holder}.

We now claim that for $0\leq s\leq t\leq \tau'$, we have that $\cEgh_{s,t}(y_{t})=y_s$. To show this, we observe that \eqref{equationholdermore1} yields $\cEgh_{t,\tau'}(y_{\tau'})=y_t$ and thus
\[
\cEgh_{s,t}(y_{t})=\cEgh_{s,t}(\cEgh_{t,\tau'}(y_{\tau'}))=\cEgh_{s,\tau'}(y_{\tau'})=y_s,
\]
where the last equality also comes from  \eqref{equationholdermore1}. We thus see that $y$ is an $\cEgh$-martingale on $[0,\tau' ]$ and thus
$\kpp_{\tau' }=0$. In particular, we have $\cEgh_{0,\tau'}(y_{\tau'})=y_0$ and thus, using \eqref{equation holder}, we obtain
\begin{equation}\label{equforrationaltime}
y_0=\cEgh_{0,\tau'}(x_{\tau'})=\cEgh_{0,\tau'}(y_{\tau'})=\cEgh_{0,\tau'}\big(V_{\tau'}(x_2-p,\psi)\big).
\end{equation}
By combining this equality with the inequality $y_{\tau'}\leq x_{\tau'}$ and the strict comparison property of $\cEgh$, we conclude that $y_{\tau'}=x_{\tau'}$. We have thus shown that if  $\tau'$ is a rational exercise time, then conditions (i)--(ii) are valid.

Let us show that $\tau^h$ is a rational exercise time. From the definition of $\tau^h$ and the right-continuity of $y$ and $x$, we infer that $y_{\tau^h}=x_{\tau^h}$.  Equality $\kpp_{\tau^h}=0$ has been already established in Proposition \ref{pro4.5}. Hence $\tau^h$ satisfies conditions (i)--(ii) and thus it is one of the holder's rational exercise times and it is the earliest one, since $y_t<x_t$ for all $t\in[0,\tau^h)$.

It remains to prove that $\bar{\tau}^h$ is the latest rational exercise time under an additional assumption that the process $\kpp$ is continuous so that $\kpp = \kpp^c$. We need to show that $y_{\bar{\tau}^h}=x_{\bar{\tau}^h}$. For an arbitrary $\varepsilon>0$, there exists $\delta\in[0,\varepsilon]$ such that $\kpp_{\bar \tau^h+\delta}>0$. Since $\int_0^T (x_t-y_t)\, d\kpp_t=0$, from the right-continuity of processes $y$ and $x$ and the inequality $x \geq y$, we obtain the equality $y_{\bar \tau^h}=x_{\bar \tau^h}$. Since, obviously, $\kpp_t=0$ for $t\in[0,\bar\tau^h)$, we also have $\kpp_{\bar\tau^h}=0$. This shows that $\bar\tau^h$ is one of the holder's rational exercise times. Moreover, it is the latest one since, if $\tau' \in \cT $ is such that $\mathbb{P} (\tau'>\bar{\tau}^h)>0$, then $\mathbb{P} (\kpp_{\tau'}>0)>0$ and thus the equality $\kpp_{\tau'}=0$ cannot hold. Note, however, that if the continuity of $\kpp$ is not postulated, then it may happen that $\kpp_{\bar \tau^h}\ne 0$ in which case $\bar{\tau}^h$ fails to be a rational exercise time (for instance, such properties are always true if $\kpp = \kpp^d$).
\endproof

\begin{remark}\label{rem4.4} {\rm
From the proof of Theorem \ref{the4.4} (see, in particular, equation (\ref{equforrationaltime})), the inequality $V_{\tau'}(x_2-p,\psi )\geq x_{\tau'}$ and the strict comparison property of $\cEgh$, we deduce that when the equality $\wh{p}^{r,h}(x_2,\cC^a)=\wh{p}^{s,h}(x_2,\cC^a)$ holds, then for any rational exercise time given by  Definition \ref{holderrationaltime} we have that $V_{\tau}(x_2-\wh{p}^{r,h},\psi)=x_{\tau}$, meaning that a rational exercise time is also a holder's break-even time. It is also obvious that a holder's break-even time is a rational exercise time. Thus when the equality $\wh{p}^{r,h}(x_2,\cC^a)=\wh{p}^{s,h}(x_2,\cC^a)$ holds, then the inequality $V_{\tau}(x_2-\wh{p}^{r,h},\psi)\geq x_{\tau}$ in Definition \ref{holderrationaltime} can be replaced by the equality $V_{\tau}(x_2-\wh{p}^{r,h},\psi)=x_{\tau}$. Once again,
it is worth noting that this observation is fully consistent with the standard definition of a rational exercise time for the holder of an American option in a classical complete, linear market model, such as the Black and Scholes model.}
\end{remark}

\section*{Acknowledgments}
The research of T. Nie and M. Rutkowski was supported by the DVC Research Bridging Support Grant {\it Non-linear Arbitrage Pricing of Multi-Agent Financial Games}. The work of T. Nie was supported by the National Natural Science Foundation of China (No. 11601285) and the Natural Science Foundation of Shandong Province (No. ZR2016AQ13).






\begin{thebibliography}{20}
{\parskip= 3 pt

\bibitem{AO2016}
Aazizi, S. and Ouknine, Y.:
Strong envelope and strong supermartingale: Application to reflected backward stochastic differential equation.
Working paper, 2016 (arXiv:1112.0255v2).

\bibitem{BB2017}
Baadi, B. and Ouknine, Y.:
Reflected BSDEs when the obstacle is not right-continuous in a general filtration.
{\it ALEA -- Latin American Journal of Probability and Mathematical Statistics} 14 (2017), 201--218.

%

\bibitem{BY2012}
Bayraktar, E. and Yao, S.:
Quadratic reflected BSDEs with unbounded obstacles.
{\it Stochastic Processes and their Applications} 122 (2012), 1155--1203.


\bibitem{B1984}
Bensoussan, A.:
On the theory of option pricing.
{\it Acta Applicandae Mathematicae} 2 (1984), 139--158.

\bibitem{BCR2018}
Bielecki, T. R., Cialenco, I., and Rutkowski, M.:
Arbitrage-free pricing of derivatives in nonlinear market models.
{\it Probability, Uncertainty and Quantitative Risk} 3/2 (2018), DOI 10.1186/s41546-018-0027-x.

\bibitem{BR2015}
Bielecki, T. R. and Rutkowski, M.:
Valuation and hedging of contracts with funding costs and collateralization.
{\it SIAM Journal on Financial Mathematics} 6 (2015), 594--655.



\bibitem{CFS2008}
Carbone, R., Ferrario, B., and Santacroce, M.:
Backward stochastic differential equations driven by c\`{a}dl\`{a}g martingales.
{\it Theory of Probability and its Applications} 52(2) (2008), 304--314.

\bibitem{CM2008}
Cr\'epey, S. and Matoussi, A.:
Reflected and doubly reflected BSDEs with jumps: a priori estimates and comparison.
{\it The Annals of Probability} 18(5) (2008), 2041--2069.



\bibitem{CK1996}
Cvitani\'c, J. and Karatzas, I.:
Backward stochastic differential equations with reflection and Dynkin games.
{\it The Annals of Probability} 24(4) (1996), 2024--2056.



\bibitem{DQS2016}
Dumitrescu, R., Quenez, M. C., and Sulem, A.:
Generalized Dynkin games and doubly reflected BSDEs with jumps.
{\it Electronic Journal of Probability} 21/64 (2016), 1--32.

\bibitem{DQS2017}
Dumitrescu, R., Quenez, M. C., and Sulem, A.:
Game options in an imperfect market with default.
{\it SIAM Journal on Financial Mathematics} 8 (2017), 532--559.

\bibitem{DQS2018}
Dumitrescu, R., Quenez, M. C., and Sulem, A.:
American options in an imperfect complete market with default.
To appear in {\it ESAIM: Proceedings and Surveys} (2018).


\bibitem{EK1981}
El Karoui, N.: Les aspects probabilistes du contr\^ole stochastique.
In {\it Lecture Notes in Mathematics 876, Ecole d'Et\'e de Probabilit\'es de Saint-Flour IX, 1979},
P.-L. Hennequin (Ed.), Springer, Berlin, 1981, pp. 73--238.

\bibitem{ELH1997}
El Karoui, N. and Huang, S. J.:
A general result of existence and uniqueness of backward stochastic differential equations.
In {\it Backward Stochastic Differential Equations, Pitman Research Notes in Mathematics
Series 364}, N. El Karoui and L. Mazliak (Eds.), Addison Wesley Longman Ltd, Harlow, Essex, 1997, pp. 27--36.

\bibitem{EQK1997}
El Karoui, N., Kapoudjian, C., Pardoux, E., Peng, S., and Quenez, M. C.:
Reflected solutions of backward SDE's, and related obstacle problems for PDE's.
{\it The Annals of Probability} 25 (1997), 702--737.

\bibitem{EPAQ1997}
El Karoui, N., Pardoux, E., and Quenez, M. C.:
Reflected backward SDEs and American options.
In {\it Numerical Methods in Finance},
L. C. G. Rogers and D. Talay (Eds.), Cambridge University Press, Cambridge, 1997, pp. 215--231.

\bibitem {EPQ1997}
El Karoui, N., Peng, S., and Quenez, M. C.:
Backward stochastic differential equations in finance.
{\it Mathematical Finance} 7 (1997), 1--71.

\bibitem{EQ1997}
El Karoui, N. and Quenez, M. C.:
Non-linear pricing theory and backward stochastic differential equations.
{\it Lecture Notes in Mathematics 1656}, B. Biais et al. (Eds.),
Springer, Berlin, 1997,  pp. 191--246.

\bibitem{ES2008}
Essaky, E. H.:
Reflected backward stochastic differential equation with jumps and RCLL obstacle.
{\it Bulletin des Sciences Math\'ematiques} 132(8) (2008), 690--710.



\bibitem{GIOOQ2017}
Grigorova, M., Imkeller, P., Offen, E., Ouknine, Y., and Quenez, M. C.:
Reflected BSDEs when the obstacle is not right-continuous and optimal stopping.
{\it The Annals of Applied Probability} 27 (2017), 3153--3188.

\bibitem{GIOQ2017}
Grigorova, M., Imkeller, P., Ouknine, Y., and Quenez, M. C.:
Optimal stopping with $f$-expectations: the irregular case.
Working paper, 2017 (hal-01403616v2).

\bibitem{GQ2017}
Grigorova, M. and Quenez, M. C.:
Optimal stopping and a non-zero-sum Dynkin game in discrete time with risk measures induced by BSDEs.
{\it Stochastics: An International Journal of Probability and Stochastic Processes} 89 (2017), 259--279.

\bibitem{HA2002}
Hamad\`ene, S.:
Reflected BSDEs with discontinuous barrier and application.
{\it Stochastics and Stochastics Reports} 74(3-4) (2002), 571--596.


\bibitem{HO2016}
Hamad\`ene, S. and Ouknine, Y.:
Reflected backward SDEs with general jumps.
{\it Theory of Probability and its Applications} 60 (2016), 263--280.

\bibitem{JLL1990}
Jaillet, P., Lamberton, D., and Lapeyre, B.:
Variational inequalities and the pricing of American options.
{\it Acta Applicandae Mathematicae} 21(3) (1990), 263--289.

\bibitem{KK2004}
Kallsen, J. and K\"uhn, C.:
Pricing derivatives of American and game type in incomplete markets.
{\it Finance and Stochastics} 8 (2004), 261--284.



\bibitem{K1988}
Karatzas, I.:
On the pricing of American options.
{\it Applied Mathematics and Optimization} 17 (1988), 37--60.


\bibitem{KK1998}
Karatzas, I. and Kou, S.:
Hedging American contingent claims with constrained portfolios.
{\it Finance and Stochastics} 2 (1998), 215--258.

\bibitem{KY2000}
Kifer, Y.: Game options.
{\it Finance and Stochastics} 4 (2000), 443--463.

\bibitem{KY2013}
Kifer, Y.:
Dynkin games and Israeli options.
{\it ISRN Probability and Statistics} (2013), ID856458, 17 pages.

\bibitem{KNR2018}
Kim, E., Nie, T., and Rutkowski, M.:
Valuation and hedging of game options in nonlinear models.
Working paper, 2018.

\bibitem{K2012}
Klimsiak, T.:
Reflected BSDEs with monotone generator.
{\it Electronic Journal of Probability} 17/107 (2012), 1--25.

\bibitem{K2015}
Klimsiak, T.:
Reflected BSDEs on filtered probability spaces.
{\it Stochastic Processes and their Applications}
125 (2015), 4204--4241.

\bibitem{KR2016}
Klimsiak, T. and Rozkosz, A.:
The early exercise premium representation for American options on multiply assets.
{\it Applied Mathematics and Optimization} 73 (2016), 99-–114.

\bibitem{KRS2016}
Klimsiak, T., Rzymowski, M., and S\l omi\'nski, L.:
Reflected BSDEs with regulated trajectories.
Working paper, 2016 (arXiv:1608.08926v1).

\bibitem{KQ2012}
Kobylanski, M. and Quenez, M. C.:
Optimal stopping time problem in a general framework.
{\it Electronic Journal of Probability} 17/72 (2012), 1--28.


%





\bibitem{M1992}
Myneni, R.:
The pricing of the American option.
{\it The Annals of Applied Probability} 2 (1992), 1--23.


\bibitem{NR1}
Nie, T. and Rutkowski, M.:
Fair bilateral prices in Bergman's model with exogenous collateralization.
{\it International Journal of Theoretical and Applied Finance} 18 (2015), 1550048. 

\bibitem{NR2}
Nie, T. and Rutkowski, M.:
BSDEs driven by multidimensional martingales and their applications to markets with funding costs.
{\it Theory of Probability and its Applications} 60 (2016), 604--630.

\bibitem{NR3}
Nie, T. and Rutkowski, M.:
A BSDE approach to fair bilateral pricing under endogenous collateralization.
{\it Finance and Stochastics} 20 (2016), 855--900.

\bibitem{NR4}
Nie, T. and Rutkowski, M.:
Fair bilateral prices under funding costs and exogenous collateralization.
{\it Mathematical Finance} 28 (2018), 621--655.





\bibitem{P2004a}
Peng, S.:
Nonlinear expectations, nonlinear evaluations and risk measures.
In {\it Lecture Notes in Mathematics 1856}, M. Frittelli and W. Runggaldier (Eds.),
Springer, Berlin, 2004, pp. 165--253.

\bibitem{P2004b}
Peng, S.:
Dynamically consistent nonlinear evaluations and expectations.
Working paper, 2004 (arXiv:0501415v1).

\bibitem{R2002}
Rogers, L. C. G.:
Monte Carlo valuation of American options.
{\it Mathematical Finance} 12(3) (2002), 271--286.


\bibitem{QS2014}
Quenez, M. C. and Sulem, A.:
Reflected BSDEs and robust optimal stopping time for dynamic risk measures with jumps.
{\it Stochastic Processes and their Applications} 124 (2014), 3031--3054.

}

\end{thebibliography}
\end{document}